\pdfoutput=1
\documentclass[11pt]{article}
\usepackage[T1]{fontenc}
\usepackage{lmodern}
\usepackage{slantsc}
\usepackage[protrusion=true,expansion=true]{microtype}
\usepackage{breakcites}
\usepackage{amsmath,amssymb,amsfonts,amsthm}
\usepackage{subcaption}
\usepackage{graphicx}
\usepackage{fullpage}
\usepackage{setspace}
\usepackage[backref=page]{hyperref}
\usepackage{color}
\usepackage{wrapfig}
\usepackage{tikz}
\usetikzlibrary{decorations.pathreplacing}
\usepackage{algorithm}
\usepackage[noend]{algpseudocode}
\usepackage[framemethod=tikz]{mdframed}
\usepackage{xspace}
\usepackage{pgfplots}
\usepackage{framed}
\usepackage{thmtools}
\usepackage{thm-restate}
\usepackage{tabu}
\usepackage{booktabs} 
\usepackage{fancyhdr}
\pgfplotsset{compat=1.5}

\newtheorem{theorem}{Theorem}[section]

\newtheorem{lemma}[theorem]{Lemma}

\newtheorem{definition}[theorem]{Definition}
\newtheorem{remark}[theorem]{Remark}

\newtheorem{fact}[theorem]{Fact}

\newenvironment{proofof}[1]{\begin{trivlist} \item {\bf Proof
#1:~~}}
  {\qed\end{trivlist}}

\newcommand{\namedref}[2]{\hyperref[#2]{#1~\ref*{#2}}}
\newcommand{\thmlab}[1]{\label{thm:#1}}
\newcommand{\thmref}[1]{\namedref{Theorem}{thm:#1}}
\newcommand{\lemlab}[1]{\label{lem:#1}}
\newcommand{\lemref}[1]{\namedref{Lemma}{lem:#1}}

\newcommand{\seclab}[1]{\label{sec:#1}}
\newcommand{\secref}[1]{\namedref{Section}{sec:#1}}
\newcommand{\applab}[1]{\label{app:#1}}
\newcommand{\appref}[1]{\namedref{Appendix}{app:#1}}
\newcommand{\factlab}[1]{\label{fact:#1}}
\newcommand{\factref}[1]{\namedref{Fact}{fact:#1}}
\newcommand{\remlab}[1]{\label{rem:#1}}
\newcommand{\remref}[1]{\namedref{Remark}{rem:#1}}
\newcommand{\figlab}[1]{\label{fig:#1}}
\newcommand{\figref}[1]{\namedref{Figure}{fig:#1}}
\newcommand{\alglab}[1]{\label{alg:#1}}
\renewcommand{\algref}[1]{\namedref{Algorithm}{alg:#1}}
\newcommand{\tablelab}[1]{\label{tab:#1}}
\newcommand{\tableref}[1]{\namedref{Table}{tab:#1}}
\newcommand{\deflab}[1]{\label{def:#1}}
\newcommand{\defref}[1]{\namedref{Definition}{def:#1}}

\newcommand{\PPr}[1]{\ensuremath{\mathbf{Pr}\left[#1\right]}}

\newcommand{\Ex}[1]{\ensuremath{\mathbb{E}\left[#1\right]}}

\renewcommand{\O}[1]{\ensuremath{\mathcal{O}\left(#1\right)}}
\newcommand{\tO}[1]{\ensuremath{\tilde{\mathcal{O}}\left(#1\right)}}
\newcommand{\eps}{\varepsilon}
\def \coreset    {\mdef{\textsc{Coreset}}}

\def \calA    {\mdef{\mathcal{A}}}

\def \calG    {\mdef{\mathcal{G}}}
\def \calH    {\mdef{\mathcal{H}}}

\def \calO    {\mdef{\mathcal{O}}}
\def \calP    {\mdef{\mathcal{P}}}
\def \calQ    {\mdef{\mathcal{Q}}}

\def \calS    {\mdef{\mathcal{S}}}

\def \calU    {\mdef{\mathcal{U}}}

\def \bA    {\mdef{\mathbf{A}}}
\def \bB    {\mdef{\mathbf{B}}}

\def \bM    {\mdef{\mathbf{M}}}

\def \bZ    {\mdef{\mathbf{Z}}}

\def \ba    {\mdef{\mathbf{a}}}
\def \bb    {\mdef{\mathbf{b}}}

\newcommand{\mdef}[1]{{\ensuremath{#1}}\xspace}  

\DeclareMathOperator*{\polylog}{polylog}
\DeclareMathOperator*{\poly}{poly}

\newcommand{\ignore}[1]{}

\newif\ifnotes\notestrue 
\ifnotes
\newcommand{\samson}[1]{\textcolor{blue}{{\bf (Samson:} {#1}{\bf ) }} 
\marginpar{\tiny\bf
             \begin{minipage}[t]{0.5in}
               \raggedright S:
            \end{minipage}}}
\newcommand{\david}[1]{\textcolor{purple}{{\bf (David:} {#1}{\bf ) }} \marginpar{\tiny\bf
             \begin{minipage}[t]{0.5in}
               \raggedright D:
            \end{minipage}}} 
\newcommand{\shenghao}[1]{\textcolor{red}{{\bf (Shenghao:} {#1}{\bf ) }} 
\marginpar{\tiny\bf
             \begin{minipage}[t]{0.5in}
               \raggedright S:
            \end{minipage}}}
\else
\newcommand{\samson}[1]{}
\newcommand{\david}[1]{}
\newcommand{\shenghao}[1]{}
\fi

\makeatletter
\renewcommand*{\@fnsymbol}[1]{\textcolor{mahogany}{\ensuremath{\ifcase#1\or *\or \dagger\or \ddagger\or
 \mathsection\or \triangledown\or \mathparagraph\or \|\or **\or \dagger\dagger
   \or \ddagger\ddagger \else\@ctrerr\fi}}}
\makeatother

\providecommand{\email}[1]{\href{mailto:#1}{\nolinkurl{#1}\xspace}}

\definecolor{mahogany}{rgb}{0.75, 0.25, 0.0}
\definecolor{darkblue}{rgb}{0.0, 0.0, 0.55}
\definecolor{darkpastelgreen}{rgb}{0.01, 0.75, 0.24}
\definecolor{darkgreen}{rgb}{0.0, 0.2, 0.13}
\definecolor{darkgoldenrod}{rgb}{0.72, 0.53, 0.04}
\definecolor{darkred}{rgb}{0.55, 0.0, 0.0}
\definecolor{forestgreenweb}{rgb}{0.13, 0.55, 0.13}
\definecolor{greencss}{rgb}{0.0, 0.5, 0.0}
\definecolor{bleudefrance}{rgb}{0.19, 0.55, 0.91}

\hypersetup{
     colorlinks   = true,
     citecolor    = mahogany,
	 linkcolor	  = darkpastelgreen
}

\fancypagestyle{pg}
{
\lhead{}
\rhead{}
\cfoot{--\ \thepage\ --}

}

\AtBeginDocument{%
  \DeclareFontShape{T1}{lmr}{m}{scit}{<->ssub*lmr/m/scsl}{}%
}
\begin{document}

\allowdisplaybreaks

\title{Nearly Space-Optimal Graph and Hypergraph Sparsification in Insertion-Only Data Streams}
\author{
Vincent Cohen-Addad\thanks{Google Research.
E-mail: \email{cohenaddad@google.com}.}
\and
David P. Woodruff\thanks{Carnegie Mellon University and Google Research. 
E-mail: \email{dwoodruf@cs.cmu.edu}. 
Supported in part by Office of Naval Research award number N000142112647 and a Simons Investigator Award.}
\and
Shenghao Xie\thanks{Texas A\&M University. 
E-mail: \email{xsh1302@gmail.com}. 
Supported in part by NSF CCF-2335411.}
\and
Samson Zhou\thanks{Texas A\&M University. 
E-mail: \email{samsonzhou@gmail.com}. 
Supported in part by NSF CCF-2335411. 
The author gratefully acknowledges funding provided by the Oak Ridge Associated Universities (ORAU) Ralph E. Powe Junior Faculty Enhancement Award.}
}
\maketitle

\begin{abstract}
We study the problem of graph and hypergraph sparsification in insertion-only data streams. The input is a hypergraph $H=(V, E, w)$ with $n$ nodes, $m$ hyperedges, and rank $r$, and the goal is to compute a hypergraph $\widehat{H}$ that preserves the energy of each vector $x \in \mathbb{R}^n$ in $H$, up to a small multiplicative error. In this paper, we give a streaming algorithm that achieves a $(1+\varepsilon)$-approximation, using $\frac{rn}{\varepsilon^2} \log^2 n \log r \cdot\text{poly}(\log \log m)$ bits of space, matching the sample complexity of the best known offline algorithm up to $\text{poly}(\log \log m)$ factors. Our approach also provides a streaming algorithm for graph sparsification that achieves a $(1+\varepsilon)$-approximation, using $\frac{n}{\varepsilon^2} \log n \cdot\text{poly}(\log\log n)$ bits of space, improving the current bound by $\log n$ factors. Furthermore, we give a space-efficient streaming algorithm for min-cut approximation. Along the way, we present an online algorithm for $(1+\varepsilon)$-hypergraph sparsification, which is optimal up to poly-logarithmic factors. As a result, we achieve $(1+\varepsilon)$-hypergraph sparsification in the sliding window model, with space optimal up to poly-logarithmic factors. Lastly, we give an adversarially robust algorithm for hypergraph sparsification using $\frac{n}{\varepsilon^2} \cdot\text{poly}(r, \log n, \log r, \log \log m)$ bits of space.
\end{abstract}

\section{Introduction}
Graphs are fundamental structures that naturally model complex relationships in real-world data, from social networks and transportation systems to knowledge graphs and human brains. 
Because of their great expressive power, these relational models are fundamental to research in computer science, data science, and machine learning, in addition to many other fields. 
\emph{Graph cuts} are used to partition graphs into distinct regions by minimizing a cost function, thereby providing insightful information. 
For example, in image segmentation, graph cuts help separate objects from the background by modeling pixels as nodes in a graph and optimizing the partitioning based on intensity or color differences~\cite{BoykovK04}. 
In network clustering, graph cuts are used to detect communities within social networks by partitioning nodes into groups with strong internal connections while minimizing connections between different groups~\cite{newman2006modularity}. 
In hierarchical clustering, sparse graph cuts are used to increasingly refine subgraphs to achieve better performance for Dasgupta's objective \cite{Dasgupta16,BravermanEWZ25,Deng0U0Z25}. 
More generally, spectral methods not only preserve key quantities such as cut sizes, but also accommodate more complex partitioning, e.g., multi-way cuts, and they consider other concepts such as the graph Laplacian and its eigenvalues.

With the explosive growth of large-scale databases and the increasing demand for scalable machine learning and AI systems, graphs have become more complex and massive than ever before. 
In many modern applications, the sheer volume of graphs creates significant computational bottlenecks, making it crucial to obtain smaller approximate graphs while preserving key features. 
\emph{Cut sparsifiers} and \emph{spectral sparsifiers} are smaller, more efficient representations of large-scale graphs that approximately preserve cut sizes and spectral properties, respectively. 
These techniques are therefore often used to accelerate graph-related algorithms and reduce computational overhead. 
For instance, graph sparsifiers have been used to enhance the performance of graph representation learning methods by reducing irrelevant information and preserving important structural features in the graph \cite{CalandrielloKLV18,ZhuSLHB19,ZengZSKP20,ZhengZCSNYC020,BravermanHMSSZ21,ChunduruZGB22,ZhangST24}, spectral sparsifiers have been used to construct lightweight graph neural networks while maintaining their overall effectiveness \cite{LiDZKY23,XieCZZWWM024}, and hypergraph sparsifiers have been used to produce efficient algorithms related to cuts and flows \cite{ChekuriX17,ChandrasekaranXY18,Chekuri018,VeldtBK23}, which have been applied to hypergraph partitioning and clustering \cite{AkbudakKA13,DeveciKC13,BallardDKS16} and machine learning tasks \cite{LiM17,LiM18,VeldtBK20}. 
In addition, graph sparsification is a key algorithmic component in theoretical computer science, e.g., computing the eigenvalues of the Laplacian matrix \cite{SomaY19} and solving a Laplacian system, which is used to produce fast symmetric diagonally dominant (SDD) system solvers \cite{KoutisMP14,SpielmanT14,SomaY19}.

\paragraph{Data streams.}
As datasets have grown dramatically in size in recent years, there is an increasing interest in large-scale computational models that avoid storing the entire dataset or making multiple passes over it. 
This has led to the development of the streaming model, which processes data sequentially under strict memory constraints, making it well-suited for handling massive datasets efficiently, e.g., network activity logs, IoT device streams, financial transactions, database event records, and real-time scientific observations. 
In this paper, we ask: 
\begin{center}
\emph{Do graph sparsification problems require additional space complexity in the streaming model, compared to the offline setting?}
\end{center}

\subsection{Our Contributions}
In this paper, we provide algorithms for cut and spectral sparsification problems, losing only poly-iterated logarithmic factors in space, i.e., $\polylog(\log n)$ factors for a graph with $n$ vertices. 

\paragraph{Graph Sparsification.}
The graph sparsification problem has been extensively studied \cite{SpielmanS08,BatsonSS14,Lee017,JambulapatiS18}, and there are well-known constructions for spectral sparsifiers with $\O{\frac{n}{\eps^2}}$ edges. 
We provide efficient algorithms for graph spectral sparsification in insertion-only streams that have optimal space complexity compared to the best sample complexity of an offline algorithm, up to poly-iterated logarithmic factors in $n$. 
Here, we represent a graph as a three-tuple $G = (V,E,w)$, where $V$ is the set of vertices with size $n$, $E$ is the set of edges with size $m = \poly(n)$, and $w: E \to [\poly(n)]$ is the weight assignment function, where the weights are positive integers upper bounded by $\poly(n)$. 
We give our formal statement as follows. 
\begin{theorem} [Streaming graph sparsification]
\thmlab{thm:graph:main}
Given a graph $G=(V,E, w)$ with $n$ vertices defined by an insertion-only stream, there exists an algorithm that gives a $(1+\eps)$-spectral sparsifier with probability $1-\frac{1}{\poly(n)}$, storing $\frac{n}{\eps^2} \poly(\log \log n)$ edges, i.e., $\frac{n}{\eps^2} \log n \poly(\log \log n)$ bits, and using $\poly(n)$ update time.
\end{theorem}

By comparison, offline constructions for graph spectral sparsifiers can be combined with the merge-and-reduce framework for coreset constructions to achieve algorithms with $\O{\frac{n}{\eps^2}\log^4 n}$ edges in the insertion-only setting.  
Additionally, \cite{CohenMP20} produces a spectral sparsifier with $\O{\frac{n}{\eps^2}\log^2 n}$ edges in the online model, and similarly, \cite{KapralovMMMNST20} achieves a spectral sparsifier with $\O{\frac{n}{\eps^2}\log n}$ edges in the dynamic streaming model, where both insertions and deletions of edges are permitted; both of these results can also be applied to the insertion-only setting. 
Here, we note that storing each edge in the sparsifier uses $\O{\log n}$ bits of space.
Our result avoids the additive $\log n$ factors in previous results, achieving a nearly-optimal space.

\paragraph{Graph min-cut approximation.}
In addition, we construct a space-efficient streaming algorithm that approximates the graph min-cut in insertion-only data streams.
\begin{theorem} [Streaming min-cut approximation]
\thmlab{thm:cut:main}
Given a graph $G =(V, E, w)$ with $n$ vertices defined by an insertion-only stream, there exists an algorithm which outputs a $(1+\eps)$-approximation to the size of the min-cut of the graph $G$ with probability $1-\frac{1}{\poly(n)}$, storing $\frac{n}{\eps} \cdot \polylog(n,\frac{1}{\eps})$ edges, i.e., $\frac{n}{\eps} \cdot \polylog(n,\frac{1}{\eps})$ bits, and using $\poly(n)$ update time. 
\end{theorem}
We remark that our result improves on the logarithmic dependence in $n$ compared to the previous work \cite{DingGLLNSW24}, while maintaining a $\poly(n)$ update time.
That is, our algorithm uses $\O{\frac{n}{\eps} \cdot \log^c(\frac{n}{\eps})\log \frac{1}{\eps}}$ bits of space, while the algorithm of \cite{DingGLLNSW24} uses $\frac{n}{\eps} \cdot \log^{c+1}(\frac{n}{\eps})$ bits of space. 

\paragraph{Hypergraph sparsification.}
More broadly, hypergraphs are a generalization of graphs where each hyperedge can connect more than $2$ nodes, enabling the representation of multi-way relationships beyond pairwise connections.
This ability to capture higher-order correlations makes hypergraphs particularly valuable in various aspects of machine learning research, such as hypergraph neural networks \cite{FengYZJG19,GaoFJJ23,FengLYG24} and hypergraph clustering \cite{TakaiMIY20}.
\cite{SomaY19} formalizes the definition of a $(1+\eps)$-hypergraph spectral sparsifier and gives a construction with $\O{\frac{n^3}{\eps^2}\log n}$ hyperedges. 
\cite{BansalST19} achieves an upper bound of $\O{\frac{n r^3}{\eps^2} \log n }$, where $r = \max_{e \in E} |e|$ is the rank of the hypergraph. 
\cite{KaplralovTY21} gives an upper bound of $\O{\frac{n}{\eps^4}\log^3 n}$. 
\cite{JambulapatiLS23} and \cite{Lee23} simultaneously improve the upper bound to $\O{\frac{n}{\eps^2}\log n \log r}$.

Note that the above results are in the offline setting and it is uncertain whether they can be adapted to our streaming setting.
In contrast, existing algorithms use $\O{\frac{n r}{\eps^2} \cdot \log^2 n \log r}$ hyperedges \cite{SomaTY24} or $\frac{n}{\eps^2}\cdot\polylog(m)$ hyperedges \cite{KhannaPS25}.
However, they are not tight in $r$ and $\log m$ factors, which is \emph{prohibitively large} for streaming large-scale hypergraphs, since $r$ could be $\Omega(n)$ and $m$ could be $\Omega(2^n)$ in the worst case.
Therefore, the goal of our work is to improve both the $r$ and $\log m$ factors in the streaming setting.
We give the first streaming algorithm that is optimal up to poly-iterated logarithmic factors in $m$, compared to the current best sample complexity upper bound in the offline setting \cite{JambulapatiLS23,Lee23}.
Again, we represent a hypergraph as a three-tuple $H = (V,E,w)$, where $V$ is the set of vertices, $E$ is the set of hyperedges with rank $r$, and $w: E \to [\poly(n)]$ is the weight assignment function.
\begin{theorem} [Streaming hypergraph sparsification]
\thmlab{thm:hyper:main}
Given a hypergraph $H=(V,E, w)$ with $n$ vertices, $m$ hyperedges, and rank $r$ defined by an insertion-only stream, there exists an algorithm that gives a $(1+\eps)$-spectral sparsifier with probability $1-\frac{1}{\poly(m)}$ storing $\frac{n}{\eps^2}\log n \cdot \poly(\log r, \log \log m)$ hyperedges, i.e., $\frac{rn}{\eps^2}\log^2 n \cdot \poly(\log r, \log \log m)$ bits, and using $\poly(n)$ update time. 
There is also an algorithm that stores $\frac{n}{\eps^2}\log n \log r\cdot \poly(\log \log m)$ hyperedges, i.e., $\frac{rn}{\eps^2}\log^2 n \log r\cdot \poly(\log \log m)$ bits, and uses exponential update time.
\end{theorem}
Here, we note that storing each hyperedge requires $\O{r \log n + \log \log m}$ bits of space (see \remref{rem:bits:hyperedge}).
Our first algorithm only loses $\polylog r$ and $\poly(\log \log m)$ factors, which are at most $\polylog n$ in the worst case, shaving off the undesirable $\poly(n)$ factors when $m = \Omega(2^n)$.
Our second algorithm further improves the $\polylog(r)$ factor and makes it tight; however, it does not have a provable polynomial update time guaranty.

Moreover, we boost the success probability of the streaming algorithm for hypergraph sparsification to $1 - \frac{\delta}{\poly(m)}$ while only losing poly-iterated logarithmic factors in $\frac{1}{\delta}$.
We apply this result to obtain adversarially robust algorithms for streaming hypergraph sparsification.

\begin{theorem} [Streaming hypergraph sparsification, high probability]
\thmlab{thm:high:prob:main}
Given a hypergraph $H=(V,E,w)$ with $n$ vertices, $m$ hyperedges,  and rank $r$ defined by an insertion-only stream, there exists an algorithm which constructs a $(1+\eps)$-spectral sparsifier with probability $1-\frac{\delta}{\poly(m)}$. 
The algorithm stores $\frac{n}{\eps^2}\log n \cdot \poly(\log r, \log \log \frac{m}{\delta})$ hyperedges, i.e., $\frac{rn}{\eps^2}\log^2 n \cdot \poly(\log r, \log \log \frac{m}{\delta})$ bits.
\end{theorem}

Along the way, we give an online sampling procedure for hypergraph spectral sparsification that nearly matches the current best sample complexity upper bound in the offline setting \cite{JambulapatiLS23,Lee23}. 
The online setting is a more restrictive setting that must immediately and irrevocably decide whether each arriving edge should be sampled or discarded, i.e., online algorithms cannot retract decisions. 
Online algorithms are particularly useful in settings where intermediate results \emph{must} be reported as the stream continues. 
They are also beneficial in practice because downstream computation can begin immediately when an item is sampled, as the item will not be removed at a later time. 
\begin{theorem} [Online hypergraph sparsification]
\thmlab{thm:ol:main}
Given a hypergraph $H=(V,E, w)$ with $n$ vertices, $m$ hyperedges,  and rank $r$ defined by an insertion-only stream, there exists an online algorithm that uses $nr \log n \poly(\log \log m)$ bits of working memory and gives a $(1+\eps)$-spectral sparsifier with probability $1-\frac{1}{\poly(m)}$ by sampling $\O{\frac{n}{\eps^2}\log n \log m \log r}$ hyperedges.
\end{theorem}

For comparison, \cite{SomaTY24} provides an online algorithm that outputs a $(1+\eps)$-hypergraph sparsifier using $\O{\frac{n r}{\eps^2} \cdot \log^2 n \log r}$ hyperedges.
However, they require $\O{n^2}$ working memory to store the sketch that defines their sampling probability.
Our algorithm improves this bound to $n \log n \poly(\log \log m)$ bits, significantly reducing memory consumption.
This result resolves Question 6.2 of \cite{SomaTY24}, which asks for improvements in space complexity.

We summarize our results in \figref{fig:results}.
Our algorithms improve on the polylogarithmic factors in space bounds in prior work and have the same asymptotic bound as the offline algorithms up to poly-iterated logarithmic factors.
We note that the update time of our algorithms is $\poly(n)$, which could be improved if we relax the space constraint.
For instance, one can apply the merge-and-reduce framework (see detailed discussion in \secref{sec:framework}) directly to the input data stream and use near-linear time offline algorithms to construct the coresets, achieving near-linear update time in the stream length.
However, these approaches incur additional $\log m$ factors in the space complexity, which is prohibitively large for many streaming applications, where space efficiency is critical due to restricted memory constraints. In contrast, our work focuses on achieving nearly-tight space usage and is well-suited for processing massive datasets.

\begin{figure}[!htb]
\centering
{
\tabulinesep=1.1mm
\begin{tabu}{|c|c|c|c|}\hline
{\bf Type of sparsifier} & {\bf Setting} & {\bf Reference} & {\bf Space}  \\\hline\hline
Graph spectral & Offline & \cite{BatsonSS14} & $\O{\frac{n}{\eps^2} \log n}$ \\\hline
Graph spectral & Online & \cite{CohenMP20} & $\O{\frac{n}{\eps^2} \log^3 n}$ \\\hline
Graph spectral & Streaming & \cite{KapralovMMMNST20} & $\O{\frac{n}{\eps^2}\log^2 n}$ \\\hline
Graph spectral & Streaming & \thmref{thm:graph:main} & $\O{\frac{n}{\eps^2}\log n \poly(\log \log n)}$ \\\hline \hline
Graph min-cut & Offline & \cite{DingGLLNSW24} & $\O{\frac{n}{\eps}\log^{c} n}$ \\\hline
Graph min-cut & Streaming & \cite{DingGLLNSW24} & $\O{\frac{n}{\eps}\log^{c+1} n}$ \\\hline
Graph min-cut& Streaming & \thmref{thm:cut:main} & $\O{\frac{n}{\eps}\log^c n \log \frac{1}{\eps}}$ \\\hline \hline
Hypergraph spectral & Offline & \cite{JambulapatiLS23,Lee23} & $\frac{rn}{\eps^2}\log^2 n\log r$ \\\hline
Hypergraph spectral & Online & \cite{KhannaPS25} & $\frac{rn}{\eps^2}\polylog(m)$\\\hline
Hypergraph spectral & Online & \thmref{thm:ol:main} & $\frac{rn}{\eps^2}\log^2 n \log m \log r$\\\hline
Hypergraph spectral & Streaming & \thmref{thm:hyper:main} & $\frac{rn}{\eps^2}\poly(\log n, \log \log m)$ \\\hline
\end{tabu}
}
\caption{Comparison of our results and the optimal results in offline, online and streaming settings. The space is measured in number of bits.}
\figlab{fig:results}
\end{figure}

\paragraph{Adversarial robustness.} 
We also obtain hypergraph sparsifiers in the adversarially robust streaming model using our high probability result. 
This model is framed as a two-player game between an adaptive adversary, who generates the input stream, and a randomized algorithm, which processes the inputs and outputs an estimate. 
Unlike the standard streaming setting, future inputs may depend on previous interactions between the adversary and the algorithm. 
Specifically, the game proceeds in $m$ rounds: in each round the adversary selects an input based on past inputs and outputs and sends to the algorithm, then the algorithm responds with a new estimate. 
This formulation captures the adaptive nature of real-world data streams and parallels adversarial robustness research in machine learning, where inputs are strategically chosen to exploit vulnerabilities in learning systems, such as adversarial examples in neural networks~\cite{GoodfellowSS14}, optimization-based evasion attacks~\cite{Carlini017,AthalyeEIK18}, adversarial training~\cite{MadryMSTV18}, and trade-offs between accuracy and robustness~\cite{TsiprasSETM19,ChenWGLJ22}. 
Importantly, adaptivity need not arise from malicious attacks: it can also result from queries from vast databases with inner correlations or from queries designed to optimize an underlying objective~\cite{HassidimKMMS20}, settings where standard streaming algorithms fail. 
Given these perspectives, the adversarially robust streaming model has attracted significant attention~\cite{BenEliezerJWY20,Ben-EliezerY20,HassidimKMMS20,AlonBDMNY21,BravermanHMSSZ21,KaplanMNS21,WoodruffZ21b,BeimelKMNSS22,Ben-EliezerEO22,Ben-EliezerJWY22,ChakrabartiGS22,AvdiukhinMYZ19,AssadiCGS23,AttiasCSS23,CherapanamjeriS23,DinurSWZ23,WoodruffZZ23a,GribelyukLWYZ24,WoodruffZ24,GribelyukLWYZ25}.

\cite{BenEliezerJWY20} introduced a framework that transforms a streaming algorithm to an adversarially robust streaming algorithm by setting the failure probability $\delta$ sufficiently small, as long as the streaming problem has a small $\eps$-flip number.
Specifically, let $\lambda$ be the $\eps$-flip number. The adversarially robust algorithm with failure probability $\delta$ requires that the failure probability of the streaming algorithm be $\frac{\delta}{\binom{m}{\lambda}n^{\O{\lambda}}}$.
Let $f$ be the targeted function and $\{x_t\}_{t=1}^T$ be an input stream.
The $\eps$-flip number is defined as the number of times $f(x_1,\ldots,x_t)$ increases by an $\eps$-fraction.
In the hypergraph sparsification problem, the graph Laplacian of the input hypergraph's associated graph has eigenvalues that increase by an $\eps$-fraction at most $\frac{n \log m}{\eps}$ times, so the $\eps$-flip number is $\frac{n \log m}{\eps}$.
Thanks to our high-probability result, applying the above framework gives us an adversarially robust algorithm with efficient space.
\begin{theorem}[Robust hypergraph sparsification]
Given a graph $H=(V,E,w)$ with $n$ vertices, $m$ hyperedges, and rank $r$ defined by an insertion-only stream, there exists an adversarially robust algorithm that constructs a $(1+\eps)$-spectral sparsifier with probability $1-\frac{\delta}{\poly(m)}$ storing $\frac{n}{\eps^2}\poly(\log n, \log r, \log \log \frac{m}{\delta})$ hyperedges and $\frac{n}{\eps^2} \cdot \poly(r, \log n, \log \log \frac{m}{\delta})$ edges in the associated graph, i.e., $\frac{n}{\eps^2} \cdot \poly(r, \log n, \log r, \log \log \frac{m}{\delta})$ bits in total.
\end{theorem}

\paragraph{The sliding window model.}
An additional application of our online algorithms is to the sliding window model. 
Recall that the streaming model defines an underlying dataset while remaining oblivious to the times at which specific data points arrive, and is thus unable to prioritize recent data over older data. 
By comparison, the more general sliding window model considers the most recent $W$ updates $\{x_{n-W+1}, \ldots, x_n \}$ to be the active data, which is crucial in situations where newer information is more relevant or accurate, such as tracking trends in financial markets or analyzing recent Census data. 
Indeed, the sliding window model outperforms the streaming model in various applications \cite{BabcockBDMW02,DatarGIM02,PapapetrouGD15,WeiLLSDW16}, especially for time-sensitive settings such as data summarization \cite{ChenNZ16,EpastoLVZ17}, social media data \cite{OsborneMMLSCIMOHJCO14}, and network monitoring \cite{CormodeM05,CormodeG08,Cormode13}.
The sliding window model is particularly useful when computations \emph{must} focus only on the most recent data. 
This is important in cases where data retention is limited by regulations. 
For example, Facebook stores user search histories for up to six months \cite{FB-data}, Google stores browser data for up to nine months \cite{google-data}, and OpenAI retains API inputs and outputs for up to 30 days~\cite{OpenAI-data}. 
The sliding window model captures the ability to expire data, in alignment with these time-based retention policies, and consequently has received significant attention in various problems~\cite{LeeT06a,BravermanO07,BravermanOZ12,BravermanGLWZ18,BravermanWZ21,WoodruffZ21b,AjtaiBJSSWZ22,JayaramWZ22,BlockiLMZ23,Cohen-AddadJYZZ25}. 
Applying our online algorithm, we utilize a framework for the sliding window model to achieve an algorithm for hypergraph sparsification:
\begin{theorem} [Streaming hypergraph sparsification in the sliding window model]
Given a hypergraph $H=(V,E, w)$ with $n$ vertices, $m$ hyperedges,  and rank $r$ defined by an insertion-only stream, there exists an algorithm that gives a $(1+\eps)$-spectral sparsifier in the sliding window model with probability $1 - \frac{1}{\poly(n)}$, storing $\frac{n}{\eps^2}\polylog(m, r)$ hyperedges, i.e., $\frac{rn}{\eps^2}\log n\polylog(m, r)$ bits.
\end{theorem}

\subsection{Concurrent and independent work.} 
We summarize several relevant results from concurrent and independent work.
\paragraph{Online algorithms.}
We remark that \cite{KhannaLP25} recently also achieved an online algorithm for hypergraph sparsification that samples $\frac{n}{\eps^2}\cdot\polylog(m)$ hyperedges. 
Their approach is based on maintaining a logarithmic number of spanners and adding arriving hyperedges to a spanner based on connectivity within the spanner and a geometrically decreasing probability. 
In fact, due to efficient subroutines for these algorithms, their update time is near-linear.
Our algorithm has a suboptimal update time, but we achieve better sample complexity and more efficient working memory. 
By comparison, our approach is based on computing a balanced solution to an optimization problem, and it is not clear whether there exists an efficient algorithm to achieve a fast update time. 
We also achieve an algorithm for online spectral sparsification with fast update time while losing $\poly(r)$ factors, which suffices for the remainder of our applications. 

\paragraph{Dynamic algorithms.}
The framework in \cite{KhannaLP25} also solves hypergraph sparsification in the fully-dynamic setting, where both insertions and deletions are allowed and time efficiency is prioritized. 
Their algorithm maintains a sparsifier with $\frac{n}{\eps^2}\cdot\polylog(m)$ hyperedges while requiring sublinear update time in $m$.
Similarly, \cite{GoranciM25} provided dynamic algorithms for hypergraph sparsification using $n r^3 \poly(\log n,\frac{1}{\eps})$ space and $r^4 \poly(\log n,\frac{1}{\eps})$ update time.
These results naturally generalize the insertion-only streaming setting.
Although they achieve better time efficiency, their space bounds lose polynomial factors in $r$ and $\log m$.
As we mentioned earlier, their sub-optimal space usage hinders the applications in the streaming setting, where the space efficiency is prioritized.
On the contrary, our algorithm reduces the $\polylog(m)$ factors to  $\poly(\log \log m)$ factors while maintaining a $\poly(n)$ update time, making significant improvements for dense hypergraphs (e.g., $m = \Omega(2^n)$).

\paragraph{Hypergraph cut sparsifier.}
In addition, we note that \cite{KhannaPS25} recently showed that in an insertion-only stream, a hypergraph cut sparsifier can be computed using $\frac{rn}{\eps^2} \cdot \polylog(n)$ bits of space, establishing an $\Omega(\log m)$ separation from hypergraph cut sparsification in dynamic streams. 
Their algorithm estimates the strength of components in the hypergraph, and whenever a component gets sufficiently large strength, it contracts the component to a single vertex to save space. 
This idea extends their result to a more general bounded-deletion setting: if the stream has at most $k$ hyperedge deletions, then $\frac{rn}{\eps^2}\log k \polylog (n)$ bits of space suffice for hypergraph cut sparsification. 
By comparison, our approach is based on a streaming framework that implements merge-and-reduce coreset construction on the output of an online algorithm. 
We achieve an algorithm that solves hypergraph spectral sparsification in insertion-only streams using $\frac{rn}{\eps^2} \log^2 n \log r \poly (\log \log (m))$ bits of space, which generalizes the hypergraph cut sparsification result. 
Assuming there are no multi-hyperedges in the hypergraph, i.e., $m = \O{2^n}$, our algorithm uses $\frac{rn}{\eps^2}\log r \polylog (n)$ bits of space, and hence it avoids extraneous dependence on $\polylog(m)$ factors, i.e., we have the same dependencies as \cite{KhannaPS25} in the $\log m$ factors.
In addition, in terms of the dependence on $\log n$, our space bound has a factor of $\tO{\log^2 n}$, as opposed to $\polylog(n)$ in \cite{KhannaPS25}, and thus we achieve a better space complexity when $m$ is relatively small.

\section{Preliminaries}
In this section, we introduce several important definitions and techniques related to graphs and hypergraphs.
For an integer $n>0$, we use the notation $[n]$ to denote the set $\{1,\ldots,n\}$. 
We use $\poly(n)$ to denote a fixed polynomial of $n$, whose degree can be determined by setting parameters in the algorithm accordingly. 
We use $\polylog(n)$ to denote $\poly(\log n)$. 
We use $\tO{F} = F \polylog F$ to hide the polylogarithmic factors.

\paragraph{Hypergraphs.}
We state the formal definition of the hypergraph energy function and the hypergraph spectral sparsifier. 
\begin{definition} [Energy]
For a hypergraph $H=(V,E,w)$, we define $Q_H:\mathbb{R}^n\to\mathbb{R}$ to be the following analog of the Laplacian quadratic form $Q_H(x)=\sum_{e\in E}w(e)\cdot\max_{u,v\in e}(x_u-x_v)^2$.
We define the energy $Q_e(x)$ of a hyperedge $e$ to be $Q_e(x) = w(e)\cdot\max_{u,v\in e}(x_u-x_v)^2$, such that $Q_H(x)=\sum_{e\in E} Q_e(x)$.
\end{definition}
The above definition carries a physical interpretation as the potential energy of a family of rubber bands.
One can consider a hyperedge $e$ as a rubber band stretched to encircle several locations $\{x_v, v \in e\}$, so $\max_{u,v\in e}(x_u-x_v)^2$ is proportional to its potential energy \cite{Lee23}.
\cite{SomaY19} introduces the following notion of the hypergraph spectral sparsifier, generalizing the definition for graphs \cite{SpielmanS08}. 
The hypergraph spectral sparsifier is defined as a re-weighted subset $\widehat{H}$ of the original hypergraph $H$, which preserves the energy of $x$ on $H$ up to a $(1+\eps)$-approximation for all vectors $x \in \mathbb{R}^n$.

\begin{definition}[Hypergraph spectral sparsifier]
A weighted hypergraph $\widehat{H}$ is a $(1+\eps)$-multiplicative spectral sparsifier for $H$ if
\begin{align} \label{eq:hg:sparsifier}
    |Q_{\widehat{H}}(x) - Q_H(x)| \leq \eps \cdot Q_H(x), ~~~ \forall x \in \mathbb{R}^n
\end{align}
\end{definition}

A natural way to relate a hypergraph to a multi-graph is to replace a hyperedge with $r$ nodes by a clique of $\binom{r}{2}$ edges formed by the nodes in that hyperedge, which is stated in the following definition of associated graph.
\begin{definition}[Associated graph]
\deflab{def:associated}
Given a hypergraph $H$, we define a multi-graph $G=(V,E)$ to be the associated graph of $H$ by replacing each hyperedge $e=(u_1,\ldots,u_r)\in H$ with the $\binom{r}{2}$ edges $(u_i,u_j)$, where $1\le i<j\le r$, each with weight $w(e)$. 
\end{definition}

\paragraph{Graph Laplacian.}
We review the graph Laplacian that encodes the structure of a graph into a matrix. 
The graph Laplacian is an $n \times n$ matrix where its diagonal stores the weighted degree for each node, and the other entries store the weight of each edge $(i,j)$ in $G$. The graph Laplacian provides a way to view the graph spectral sparsifier problem as a matrix spectral approximation problem so that one can implement numerical linear algebra methods to solve it.
\begin{definition}[Graph Laplacian]
Given a weighted graph $G=(V,E, w)$ with $n$ vertices, the graph Laplacian matrix $L_G\in\mathbb{R}^{n\times n}$ is defined as
\[L_G=\sum_{e=(u,v)\in E}w(e)\cdot(\chi_u-\chi_v)^\top(\chi_u-\chi_v),\]
where $\chi_i$ denotes the elementary row vector in $\mathbb{R}^n$ with a single nonzero entry in the $i$-th coordinate. 
\end{definition}

The graph Laplacian can be written as the Gram matrix $\bA^\top \bA$, where the incidence matrix $\bA$ consists of binary vectors representing each weighted row in graph $G$. 
\begin{fact}
The graph Laplacian $L_G$ has the following properties:
\begin{enumerate}
\item 
Let $L_{uv}$ be the graph Laplacian for an edge $uv$, so that
\[L_{uv} = w(e)\cdot(\chi_u-\chi_v)^\top(\chi_u-\chi_v).\]
Then $L_G=\sum_{e=(u,v)\in E} L_{uv}$.
\item
For each $i,j\in[n]$, the $(i,j)$-th entry of $L_G$ is the weighted degree $\deg(i)=\sum_{(i,j)\in E: u\in V}w(ij)$ for $i=j$ and $-w(ij)$ for $i\neq j$, where $w(ij)$ is the weight of edge $(i,j)$. 
\item
Let $\bA\in\mathbb{R}^{|E|\times n}$ be the incidence matrix of $G$, where each row $\ba_i$ of $\bA$ corresponds to an edge $e=(u,v)\in E$, so that $\ba_{uv} := \ba_i=\sqrt{w(e)}\cdot(\chi_u-\chi_v)$. 
Then $L_G=\bA^\top\bA$.  
\end{enumerate}
\end{fact}
We remark that the matrix $A$ sometimes denotes the adjacency matrix, an $n\times n$ square matrix with $A_{uv} = w(e)$ if $(u,v) \in E$, in some other definitions of the graph Laplacian, e.g., $L = D-A$.
Here, we stick to the notion of the incidence matrix $\bA$ in \cite{JambulapatiLS23}, which is \emph{different} from the definition of the adjacency matrix.
In addition, we sometimes abuse the subscript in this paper, e.g., $\ba_i = \ba_{uv}$ if $e = (u,v)$ is represented by the $i$-th row in the incidence matrix $\bA$.

In the graph spectral sparsification problem, the energy is defined to be $x^\top L_G x$, which is $x^\top \bA^\top \bA x$. Thus, sampling the edges of the graph is equivalent to sampling rows from the matrix $\bA$. In the hypergraph spectral sparsification problem, the definition of energy is slightly different. However, we can still relate it to the graph Laplacian of its associated graph as follows. 
\begin{fact}
Let $H=(V,E,w)$ be a hypergraph, let $G=(V,F,w)$ be its associated graph, and let $L_{uv}$ be the graph Laplacian for edge $(u,v)$ in graph $G$. The energy $Q_e(x)$ of a hyperedge $e$ satisfies
\[Q_e(x) = \max_{u,v \in e} x^\top L_{uv} x\]
\end{fact}

\cite{SpielmanS08} introduces the following definition of the effective resistance of an edge in the graph, and they sample the edge with probability proportional to its effective resistance, which results in sample complexity $\O{\frac{n}{\eps^2}\log n}$. 
Due to the equivalence mentioned above between spectral sparsification and matrix spectral approximation, the effective resistance of an edge $e$ turns out to be the leverage score of the row $\ba_i$ in the incidence matrix $\bA$, where $\ba_i$ is the row representing $e$.
\begin{lemma}[Effective resistance and leverage score]
\lemlab{lem:er:lev:score}
For a graph $G=(V,E,w)$, the effective resistance of an edge $e=(u,v)\in E$ is the quantity $r_e=w(e)\cdot (\chi_u-\chi_v)L_G^{-1}(\chi_u-\chi_v)^\top$, where $\chi_i$ denotes the elementary row vector with a single non-zero entry in the $i$-th coordinate. 
Given a matrix $\bA\in\mathbb{R}^{n\times d}$, the leverage score of row $\ba_i$ is the quantity $\tau_i = \ba_i(\bA^\top\bA)^{-1}\ba_i^\top$. 
Let $\bA$ be the incidence matrix of $G$. Let $\ba_i$ be the row that represents the edge $e$. Then, we have $r_e = \tau_i$.
\end{lemma}

\paragraph{Online leverage scores.}
We introduce a core technical tool given by \cite{CohenMP20}, which gives the matrix spectral approximation in the online setting, i.e., the rows arrive sequentially in a stream. When a row $\ba_i$ arrives, we only have access to matrix $\bA_i$, which is a part of the matrix $\bA$ that arrives before $\ba_i$. Therefore, we consider the following online variation of leverage scores to approximate the real leverage scores.
\begin{definition}[Online leverage scores]
Given a matrix $\bA\in\mathbb{R}^{n\times d}$, let $\bA_i$ denote the first $i$ rows of $\bA$, i.e., $\bA_i=\ba_1\circ\ldots\circ\ba_i$, where $\circ$ denotes the row-wise concatenation operator. 
The online leverage score of row $\ba_i$ is the quantity $\ba_i(\bA_i^\top\bA_i)^{-1}\ba_i^\top$. 
\end{definition}

The next statement shows that the online leverage score is in fact an overestimate of the leverage score of $\ba_i$ in the whole matrix $\bA$.
\begin{lemma}[Monotonicity of leverage scores]
\cite{BravermanDMMUWZ20,CohenMP20,WoodruffY23}
\lemlab{lem:lev:score:monotone}
Given a matrix $\bA\in\mathbb{R}^{m\times n}$, let $\tau(\ba_i)$ denote the leverage score of $\ba_i$ and let $\tau^{\mathsf{OL}}(\ba_i)$ denote the online leverage score of row $\ba_i$. 
Then $\tau^{\mathsf{OL}}(\ba_i)\ge\tau(\ba_i)$. 
\end{lemma}

The next statement bounds the sum of online leverage scores, which bounds the sample complexity.
\begin{theorem}[Sum of online leverage scores]
\thmlab{thm:sum:ol}
\cite{CohenMP20}
Given a matrix $\bA\in\mathbb{R}^{m\times n}$, let $\tau^{\mathsf{OL}}(\ba_i)$ denote the online leverage score of row $\ba_i$. 
Then $\sum_{i=1}^m\tau^{\mathsf{OL}}(\ba_i)=\O{n\log\kappa}$, where $\kappa=\|\bA\|_2\cdot\max_{i\in[n]}\|\bA_i^{-1}\|_2$ is the online condition number of $\bA$.
\end{theorem}

With the above properties of monotonicity and bounded sum, one can generalize the following result for the online matrix spectral approximation.
\begin{theorem}[Online row sampling]
\thmlab{thm:ol:row}
\cite{CohenMP20}
Given a matrix $\bA\in\mathbb{R}^{m\times n}$ defined by an insertion-only stream, let $\tau^{\mathsf{OL}}(\ba_i)$ denote the online leverage score of row $\ba_i$. 
Then sampling $\O{n \log n \log \kappa}$ rows with probability proportional to $\tau^{\mathsf{OL}}(\ba_i)$ gives a $2$-spectral approximation to $\bA$ at all times.
\end{theorem}

\section{Online Hypergraph Spectral Sparsifier}
\seclab{sec:ol}
We start with hypergraph sparsification as it generalizes the notion of cut and graph sparsification. 
A natural starting point to sample hyperedges from a hypergraph is to consider its associated graph and the effective resistance of edges that belong to the clique defined by the hyperedges. 
\cite{BansalST19} adopts this idea and sets the sampling probability as $p_e \propto \max_{u,v \in e} r_{uv}$, where $r_{uv}$ is the effective resistance of edge $(u,v)$ in the associated graph.
However, this loses a $\poly(r)$ factor in the sample complexity. 
To avoid this, \cite{KaplralovTY21} introduces a novel weight assignment scheme for the edges in the associated graph. 
For a hyperedge $e$ with weight $w(e)$ in the hypergraph, they assign a weight $z_{uv}$ to each $u,v \in e$ in the associated graph so that they satisfy $\sum_{u,v \in e} z_{uv} = w(e)$, that is, all edges in the clique defined by $e$ sum up to $w(e)$.

Let $\mathbf{A}$ be the incidence matrix of the associated graph, and let $\mathbf{Z}$ denote the diagonal weight matrix, where its $(i,i)$-th entry is the weight $z_i$ of the edge represented by the $i$-th row in matrix $\mathbf{A}$.
They define the effective resistance of an edge $i$ in the re-weighted associated graph to be $\frac{\tau_i(\bZ^{1/2}\bA)}{z_i}$, where $\tau_i(\bZ^{1/2}\bA)$ is the leverage score of the $i$-th row in the weighted incidence matrix $\bZ^{1/2}\bA$. 
Now, $\bZ$ is chosen such that for all edges $i$ in the clique of hyperedge $e$, the ratios $\frac{\tau_i(\bZ^{1/2}\bA)}{z_i}$ are within a factor $\gamma$ of each other, where $\gamma$ is a constant. 
The intuition for this assignment is to ``balance'' the effective resistance of each edge in the clique. 
\cite{KaplralovTY21} shows that sampling $e$ with probability proportional to the $\gamma$-balanced ratios gives an improved sample complexity.

\cite{JambulapatiLS23} states a more general definition of weight assignment and sampling probability, which is called the ``group leverage score overestimate'', and still gives a valid spectral sparsifier. Indeed, for any weight assignment with $\sum_{u,v \in e} z_{uv} = w(e)$, if we sample each hyperedge $e$ with probability higher than the maximum of all ratios $\frac{\tau_i(\bZ^{1/2}\bA)}{z_i}$ of edges $i$ in the clique of $e$, the correctness of the sparsifier is guaranteed. 
Then, we are left to choose a proper $\bZ$ that gives the desired sample complexity.

\begin{definition}[Sampling probability for each hyperedge]
\deflab{def:sampling:prob}
Given a hypergraph $H= (V, E, w)$ and its associated graph $G = (V, F, z)$, where the weight assignment $z$ satisfies $\sum_{u,v \in e} z_{uv} = w(e)$ for all $e \in E$, let $\bA$ be the corresponding incidence matrix for $G$. Let $\bZ$ be the weight matrix. We set the sampling probability of $e$ to satisfy $p_e \ge w(e) \cdot \max_{u, v \in e} \ba_{uv} ( \bA^\top \bZ \bA)^{-1} \ba_{uv}^\top$.
Here, we have $\ba_{uv} ( \bA^\top \bZ \bA)^{-1} \ba_{uv}^\top=\frac{\tau_{uv}(\bZ^{1/2}\bA)}{z_{uv}}$.
We use $q_{uv}$ to denote this ratio in the following sections.
\end{definition}

\cite{JambulapatiLS23} applies a refined chaining argument by leveraging Talagrand’s growth function framework to achieve a tighter bound. 
They construct a sequence of functions that satisfy specific growth conditions with certain metrics, allowing them to control the sum of distances across different scales. 
Their core idea is to amortize the contributions of different scales rather than treating them independently, reducing unnecessary factors in the final bound. 
Their analysis can be applied as a black-box if the sampling probabilities satisfy the condition in \defref{def:sampling:prob}: (1) They are an overestimate; (2) Their values are defined independently of whether other hyperedges are sampled.

\subsection{Online Algorithm}

In the online setting, a technical challenge is that we can only define the weights $\bZ_t$ based on the hyperedges that arrived before time $t$ and their associated graph with incidence matrix $\bA_t$. However, the sampling probabilities must satisfy \defref{def:sampling:prob} for the complete hypergraph and its corresponding associated graph with incidence matrix $\bA$. 

We determine the weight assignment in an online manner. That is, when the local weight $\bZ_t$ at time $t$ is defined, we fix it when the next hyperedge $e$ arrives and add the new weights $\{z_{uv}~|~ u,v \in e\}$ to $\bZ_t$ to define $\bZ_{t+1}$. 
Thus, the matrix $\bZ_t^{1/2} \bA_t$ is consistent along the way, and hence $\frac{\tau_i(\bZ_t^{1/2}\bA_t)}{z_i}$ is an overestimate of $\frac{\tau_i(\bZ^{1/2}\bA)}{z_i}$ due to the monotonicity of online leverage scores. 
Then, we define the sampling probability $p_{e_t}$ as satisfying
$p_{e_t} \ge w(e) \cdot \max_{u, v \in e} \ba_{uv} ( \bA_t^\top \bZ_t \bA_t)^{-1} \ba_{uv}^\top$.
Therefore, the sampling probabilities $p_{e_t}$ satisfy \defref{def:sampling:prob}, and so using the chaining argument in \cite{JambulapatiLS23} gives the correctness of our algorithm.

However, the above process requires us to store the entire matrix $\bZ^{1/2} \bA$, which uses a prohibitively large working memory.
To solve this problem, we propose a local weight-assignment subroutine with reduced working memory based on online row sampling \cite{CohenMP20}. 
Let $\bB$ be a matrix, and let $\bB_i$ denote the matrix formed by the first $i$ rows of $\bB$. 
In online row sampling, we sample the newly arrived row $\bb_i$ with probability $\Bar{p_i} \propto \bb_i (\widetilde{\bB}_{i-1}^\top\widetilde{\bB}_{i-1})^{-1} \bb_i^\top$, where the prefix matrix $\widetilde{\bB}_{i-1}$ consists of the rows sampled from previous steps, and we add $\bb_i / \Bar{p_i}$ to $\widetilde{\bB}_{i-1}$ if it is sampled. 
We adapt this idea to our online sampling procedure in an iterative way. 
Suppose that we have a re-weighted incidence matrix $\bM$ that is a $2$-approximate spectral approximation of $\bZ_t^{1/2} \bA_t$. We use $\bM$ to determine the weights $z_{t+1}$ of the newly arrived hyperedge $e_{t+1}$ by calling $\textsc{GetWeightAssignment}(\bM, e)$, which is a local version of the balanced-weight-assignment procedure in \cite{KaplralovTY21} ((see \thmref{thm:balance:weight:prev}).
Then we sample the weighted edge vector $\ba_{uv} \cdot \sqrt{z_{t+1,uv}}$ by online row sampling and add to $\bM$, so the resulting matrix is still a $2$-approximate spectral approximation to the matrix $\bZ_{t+1}^{1/2} \bA_{t+1}$, which suffices for our purpose.
Our algorithm is displayed in \algref{alg:online:hg:sparsifier:opt}.
\begin{algorithm}[!htb]
\caption{Online Hyperedge Spectral Sparsifier}
\alglab{alg:online:hg:sparsifier:opt}
\begin{algorithmic}[1]
\State{{\bf Require:} Stream of $m$ hyperedges for hypergraph $H$ with rank $r$}
\State{{\bf Ensure:} Spectral sparsifier $\widehat{H}$ for $H$}
\State{$\widehat{H}\gets\emptyset$, $\rho\gets\O{\frac{1}{\eps^2}\log m \log r}$, $\bM \gets \emptyset$}
\For{hyperedge $e_t$}
\State{$z_t \gets z_{t-1} \cup \textsc{GetWeightAssignment}(\bM, e_t)$}
\Comment{See \thmref{thm:balance:weight:prev}}
\For{$u,v\in e_t$}
\State{Sample weighted row $\ba_{uv}  \cdot \sqrt{z_{t,uv}}$ to $\bM$ by online row sampling \cite{CohenMP20}}
\State{}\Comment{$\bM$ is a $2$-spectral approximation to $\bZ_t^{1/2} \bA_t$ at all times $t$, i.e., $\frac{1}{2}\cdot \bM^\top\bM \preceq \bA_t^\top \bZ_t \bA_t \preceq 2 \cdot \bM^\top \bM$}
\EndFor
\For{$u,v\in e_t$}
\State{$q_{uv}\gets \ba_{uv} \cdot w(e_t) \cdot (\bM^\top \bM)^{-1} \cdot \ba_{uv}^\top$}
\EndFor
\State{$p_{e_t}\gets \min \{1, 2\rho\cdot\max_{u,v\in e_t}q_{uv} \}$}
\State{With probability $p_{e_t}$, $\widehat{H}\gets\widehat{H}\cup\frac{1}{p_{e_t}}\cdot e_t$}
\State{{\bf Return} $\widehat{H}$}
\EndFor
\end{algorithmic}
\end{algorithm}

We next show the correctness and bound the sample complexity of our algorithm.
To begin with, we state the result of the chaining argument from \cite{JambulapatiLS23}.

\begin{theorem}[Correctness, see Theorem 10 in \cite{JambulapatiLS23}]
\thmlab{thm:correctness}
Given a hypergraph $H= (V, E, w)$ and its associated graph $G = (V, F)$ with valid weight assignment $z$, let $p_e$ be chosen independently according to \defref{def:sampling:prob}. Let $\rho = \O{\frac{1}{\eps^2}\log m \log r}$. Suppose that we sample each hyperedge $e$ in $H$ with probability $\widehat{p_e} = \min \{1, p_e\cdot \rho\}$ and scale by $\frac{1}{\widehat{p_e}}$ if $e$ is sampled. Then, the resulting hypergraph $\calH$ is a $(1+\eps)$-spectral sparsifier for $H$ with probability $1-\frac{1}{\poly(m)}$.
\end{theorem}

Next, we assume that there is a way to define the weight matrix $\bZ_t$ locally such that the sampling probability $p_{e_t}$ defined by $p_{e_t} \ge w(e) \cdot \max_{u, v \in e} \ba_{uv} ( \bA_t^\top \bZ_t \bA_t)^{-1} \ba_{uv}^\top$ satisfies \defref{def:sampling:prob} and is smaller than the sum of online leverage scores of $\bZ^{1/2} \bA$.
We state the result as follows and defer the discussion to \secref{sec:exist:bw}.
\begin{theorem}
\thmlab{thm:balance:weight:prev}
Given a graph $G=(V, F)$ with weight assignment $\bZ$ and incidence matrix $\bA$ and a hyperedge $e \subset V$ with weight $w(e)$, then there is a procedure $\textsc{GetWeightAssignment}(\bZ^{1/2}\bA, e)$ that assigns a weight $z_{uv}$ to each edge $u, v \in e$ such that it satisfies (1) $\sum_{u,v \in e} z_{uv} = w(e)$, and (2) $\max_{u,v \in e} \ba_{uv} (\bA^\top \bZ \bA)^{-1} \ba_{uv}^\top = \O{\sum_{u,v \in e} \tau_{uv}(\bZ^{1/2} \bA)}$, where $\tau$ is the leverage score function.
\end{theorem}
Then, by \thmref{thm:correctness} and \thmref{thm:balance:weight:prev}, sampling each hyperedge $e$ with probability defined by $p_{e_t} \ge w(e) \cdot \max_{u, v \in e} \ba_{uv} ( \bA_t^\top \bZ_t \bA_t)^{-1} \ba_{uv}^\top$ gives a valid sparsifier with desired sample complexity.
Recall that this approach implies a suboptimal working memory.
We apply the online row sampling scheme to sample from $\bZ^{1/2} \bA$, reducing the rows we store.

Notice that the chaining argument in \cite{JambulapatiLS23} requires that the sampling probabilities be assigned independently. 
The following statement shows that $p_{e_t}$ is defined independently of whether the previous hyperedges are sampled.
\begin{lemma}
\lemlab{lem:indep}
For each time $t$, the sampling probability $p_{e_t}$ is independent of the hyperedges that have been sampled previously. 
\end{lemma}
\begin{proof}
First, we stress that the construction of $\bM$ by online row sampling from $\bZ^{1/2}\bA$ and sampling the hyperedges are separate procedures with independent inner randomness.
Consider a fixed stream of hyperedges $e_1, \cdots, e_m$.
Let $\bM_t$ denote the matrix $\bM$ at time $t$.
When the inner randomness of online row sampling procedure is fixed, the sequence of matrices $\bM_1, \cdots, \bM_m$ is also fixed.
Then, the sampling probabilities $p_{e_t}$ are defined independently, each based on the value of $\max_{u,v \in e} \ba_{uv} \cdot w(e_t) \cdot (\bM^\top \bM)^{-1} \cdot \ba_{uv}^\top$.
Thus, it is independent of the hyperedges that have been sampled previously.
\end{proof}

Next, we show the correctness of \algref{alg:online:hg:sparsifier:opt}.
\begin{lemma}
\lemlab{lem:correctness:ol}
\algref{alg:online:hg:sparsifier:opt} outputs a $(1+\eps)$-hypergraph spectral sparsifier with probability $1-\frac{1}{\poly(m)}$.
\end{lemma}
\begin{proof}
By the guarantee of \thmref{thm:ol:row}, $\bM$ is a $2$-spectral approximation to $\bZ_t^{1/2} \bA_t$ at all times $t$. Then, the estimated quadratic form $q_{uv}$ satisfies $2q_{uv} \ge \ba_{uv} ( \bA^\top \bZ \bA)^{-1} \ba_{uv}^\top$.
In addition, by \thmref{thm:balance:weight:prev}, our weight assignment satisfies $\sum_{u,v \in e} z_{uv} = w(e)$ for each $e$, so the sampling probabilities $p_{e_t}$ satisfy \defref{def:sampling:prob}. Moreover, by \lemref{lem:indep}, the sampling probabilities are defined independently of previously sampled hyperedges. Thus, we can apply \thmref{thm:correctness} directly to show that $\widehat{H}$ is a $(1+\eps)$-hypergraph spectral sparsifier with probability $1-\frac{1}{\poly(m)}$.
\end{proof}

Next, we upper bound the sample complexity and the working memory, where the latter is the number of rows that we store in $\bM$ to compute the sampling probabilities of the hyperedges.
\begin{lemma}
\lemlab{lem:sample:ol}
With probability $1-\frac{1}{\poly(m)}$, \algref{alg:online:hg:sparsifier:opt} samples $\O{\frac{n}{\eps^2}\log n \log m \log r}$ hyperedges.   
Moreover, it uses $\O{n \log^2 n \log m}$ bits of working memory.
\end{lemma}
\begin{proof}
To bound the sample complexity, we only need to bound $\sum_{t=1}^m p_{e_t}$, where $p_{e_t}\le \rho\cdot\max_{u,v\in e_t}q_{uv}$ due to our definition. 
Notice that the local weight assignment $\bZ_{t+1}$ is defined by the sparsified graph with weighted incidence matrix $\bM$ at time $t$, then by \thmref{thm:balance:weight:prev}, our weight assignment guarantees that $\max_{u,v\in e_t}q_{uv} = \O{\sum_{u,v \in e} \tau_{uv}(\bM)}$. Since $\bM$ is a $2$-spectral approximation to $\bZ_t^{1/2} \bA_t$ at all times $t$, we have
\begin{align*}\sum_{t=1}^m \max_{u,v\in e_t}q_{uv} = & ~ \O{\sum_{t=1}^m\sum_{u,v \in e} \tau_{uv}(\bM_t)} \\ = & ~\O{\sum_{i = 1}^{|F|} \tau_i^\mathsf{OL}(\bZ^{1/2} \bA)},\end{align*}
where $\bM_t$ is the matrix $\bM$ defined at time $t$, $|F|$ is the number of rows in the incidence matrix $\bA$, and $\tau^\mathsf{OL}$ denotes the online leverage score operator. 
Utilizing \thmref{thm:sum:ol}, we bound the sum of online leverage scores on the RHS of the above equation by $\O{n \log \kappa}$. 
Therefore, we have
\[\sum_{t=1}^m \max_{u,v\in e_t}q_{uv} = \O{n \log \kappa},\]
where $\kappa=\|\bA\|_2\cdot\max_{i\in[n]}\|\bA_i^{-1}\|_2$ is the online condition number of $\bA$. We state an upper bound on the online condition number.
\begin{fact} [see Corollary 2.4 in \cite{CohenMP20}]
\factlab{fac:kappa}
Suppose that all hyperedge weights are integers from $[\poly(n)]$, then we have $\log \kappa = \O{\log n}$.
\end{fact}
Suppose that the condition in \factref{fac:kappa} is satisfied, we have
\[\sum_{t=1}^m \max_{u,v\in e_t}q_{uv} = \O{n \log n}.\]
Then, due to our choice of overestimate parameter $\rho$, we have
\[\sum_{t=1}^m p_{e_t} = \sum_{t=1}^m \rho\cdot\max_{u,v\in e_t}q_{uv} = \O{\frac{n}{\eps^2}\log n \log m \log r}.\]
Thus, we sample $\O{\frac{n}{\eps^2}\log n \log m \log r}$ hyperedges in expectation. Note that the inner randomness of sampling each hyperedge with probability $p_e$ is independent, then by standard concentration inequalities, the number of sampled hyperedges is $\O{\frac{n}{\eps^2}\log n \log m \log r}$ with probability $1 - \frac{1}{\poly(n)}$. 

In addition, by the guarantee of online row sampling in \thmref{thm:ol:row} and \factref{fac:kappa}, it suffices to sample $\O{n \log n \log m}$ rows in $\bZ^{1/2} \bA$ to construct the $2$-approximation $\bM$. 
Here, we notice that we need a $\log m$ factor in our sample complexity as opposed to the $n \log^2 n$ rows in online graph sparsifier in \thmref{thm:ol:row}.
This is because the associated graph of the hypergraph has $\poly(m)$ edges in total, so we need a $\frac{1}{\poly(m)}$ failure probability to union bound across all edges, requiring a $\log m$ factor to apply concentration inequalities.
Since we only need to store $\bM$ to compute the sampling probabilities, \algref{alg:online:hg:sparsifier:opt} uses $\O{n \log^2 n \log m}$ bits of working memory. 
\end{proof}

Combining \lemref{lem:correctness:ol} and \lemref{lem:sample:ol}, we have the following result for hypergraph spectral sparsification.
\begin{theorem}
[Online hypergraph spectral sparsifier]
\thmlab{thm:hg:online}
Given a hypergraph $H=(V,E,w)$ with $n$ vertices, $m$ hyperedges, and rank $r$, there exists an online algorithm with $\O{n \log n \log^2m}$ bits of working memory that constructs a $(1+\eps)$-spectral sparsifier with probability $1-\frac{1}{\poly(m)}$ by sampling $\O{\frac{n}{\eps^2}\log n \log m \log r}$ hyperedges. 
\end{theorem}

\subsection{Existence of Online \texorpdfstring{$(\gamma, e)$}{Gamma} Balanced Weight Assignment}
\seclab{sec:exist:bw}
In this section, we show how to construct the local weight assignment in \thmref{thm:balance:weight:prev}.
In the offline setting, the ``balancing weight'' method mentioned previously is used to upper bound the number of samples. Suppose that for a hyperedge $e$, all edges $u,v \in e$ have the same ratio $\frac{\tau_i(\bZ^{1/2}\bA)}{z_i}$. Then,
\[p_e = w(e) \cdot \ba_{uv} ( \bA^\top \bZ \bA)^{-1} \ba_{uv}^\top = w(e) \cdot \frac{\tau_{uv}(\bZ^{1/2}\bA)}{z_{uv}}.\]
Here, $(u,v)$ can be any edge in $e$.
Hence, we have $z_{uv} \cdot p_e =  w(e) \cdot \tau_{uv}$. Recall that our weight assignment satisfies $\sum_{u,v \in e} z_{uv} = w(e)$, then summing up the previous equation for all $u,v \in e$ gives $p_e = \sum_{u,v \in e} \tau_{uv}(\bZ^{1/2} \bA)$, which is exactly the sum of the leverage scores as we desired.
Since exact balanced weights are hard to find, we turn to the definition of $\gamma$-balanced weight assignment, where the ratios $\frac{\tau_i(\bZ^{1/2}\bA)}{z_i}$ are within a factor of $\gamma$ of each other.
This only incurs an additive constant factor in the sample complexity.

We extend this definition to the online setting. Now, we have a weight matrix $\bZ$, an incidence matrix $\bA$, and an incoming hyperedge $e$. We assign weights to each edge in $e$ such that they satisfy the definition of $\gamma$-balanced weights.
\begin{definition}[Online $(\gamma,e)$-balanced weight assignment]
\deflab{def:ol;gamma:bal}
Given a weighted graph $G=(V, F, z)$ and a hyperedge $e \subset V$ with weight $w(e)$, we assign a weight $z_{uv}$ to each edge $u, v \in e$ such that $\sum_{u,v \in e} z_{uv} = w(e)$ and add $(u, v)$ to $G$. We call it an online $(\gamma,e)$-balanced weight assignment if it satisfies
\[\gamma \cdot \min _{u, v \in e: z_{uv}>0} \frac{\tau_{uv}(\bZ^{1/2}\bA)}{z_{uv}} \geq \max _{u', v' \in e} \frac{\tau_{u'v'}(\bZ^{1/2}\bA)}{z_{u'v'}}.\]
\end{definition}

\cite{KaplralovTY21} provides a greedy algorithm that shifts the weights from the edges with a higher ratio to the edges with a lower ratio. They prove that such weight shift operations always increase the spanning tree potential (ST-potential) of the graph by a certain amount, which is defined as follows. 
\[\Psi(G)=\log \left[\sum_{T \in \mathbb{T}(G)} \prod_{u,v \in T} z_{uv} \right].\]
Here, $\mathbb{T}(G)$ is the set of all spanning trees of $G$.
Since the ST-potential is upper bounded, this process terminates in finite steps and results in a valid set of weights. We demonstrate that the greedy algorithm can be applied in the online setting by showing that weight shift operations still increase the ST-potential. 
Now, we formalize the definition of the weight shift operation. 
\begin{definition}[Weight shift]
Given a weighted graph $G=(V, F, z)$, an edge $(u,v) \in F$, and a weight shift factor $\lambda \in \mathbb{R}$, the graph $G + \lambda\cdot uv$ is the weighted graph $G'=(V, F, z')$ such that $z'_{uv} = \max\{0, z_{uv} + \lambda\}$ and $z'_{u'v'} = z_{u'v'}$ for all $(u',v') \in F \backslash \{(u,v)\}$.
\end{definition}

\begin{algorithm}[!htb]
\caption{Online $(\gamma, e)$-balanced weight assignment}
\alglab{alg:online:bw}
\begin{algorithmic}[1]
\State{{\bf Require:} Given weighted graph $G=(V, F, w)$, hyperedge $e \subset V$ with weight $w(e)$}
\State{{\bf Ensure:} Online $(\gamma, e)$-balanced weight assignment}
\State{{\bf Initialize:} for all $(u,v) \in e$, set $z_{uv} = w(e) / |e|$}
\State{$G \gets (V, F \cup \bigcup_{u,v \in e} (u,v), z)$}
\State{{\bf While} G is not online $(\gamma, e)$-balanced weight assignment {\bf do}}
\State{\hspace{1.5em} Select $u, v \neq u', v' \in e$ such that $q_{uv} > \gamma \cdot q_{u'v'}$ and $z_{u'v'} > 0$}
\Comment{$q_{uv}$ is the ratio in \defref{def:sampling:prob}}
\State{\hspace{1.5em} $\lambda \gets \min \{ z_{uv}, \frac{\gamma - 1}{2\gamma\cdot q_{uv}} \}$}
\State{\hspace{1.5em} $z_{uv} \gets z_{uv} + \lambda$}
\State{\hspace{1.5em} $z_{u'v'} \gets z_{u'v'} - \lambda$}
\State{{\bf Return} $G$}
\end{algorithmic}
\end{algorithm}

The following lemma upper bounds the ratio $q_{uv}$ (see \defref{def:sampling:prob}) of a bridge in a graph, which is later used to prove that we never set the weight of a bridge to $0$; thus, the connectivity ensures that the ST-potential is always well-defined.
\begin{lemma}[Upper bound, see Fact 2.8 in \cite{KaplralovTY21}]
\lemlab{lem:ub:ratio}
For any weighted graph $G=(V, F, z)$ and any edge $(u,v) \in F$, we have $z_{uv} \cdot q_{uv} \leq 1$, with equality if and only if $(u,v)$ is a bridge.
\end{lemma}

The following lemma states the increment in the ST-potential when we operate a weight shift.
\begin{lemma}[Reduction lemma, see Lemma 5.7 in \cite{KaplralovTY21}]
\lemlab{lem:reduction}
Given a weighted graph $G=(V, F, z)$ and real number $\gamma>1$, let $(u,v) \neq (u',v')$ be two edges in $F$ such that $q_{uv} >\gamma \cdot q_{u'v'}$. Then for any $\lambda \leq z_{uv}$, shifting $\lambda$ weight from $(u,v)$ to $(u',v')$ results in an increment of at least
$\log \left(1+\lambda \gamma \cdot q_{uv}-\lambda \cdot q_{uv}-\lambda^2 \gamma \cdot q_{uv}^2\right)$
of the ST-potential of $G$.
\end{lemma}

With the above results, we show that the greedy algorithm terminates in a finite number of steps.
\begin{theorem}
\algref{alg:online:bw} terminates in a finite number of steps.
\end{theorem}
\begin{proof}
This proof follows closely from the proof of Theorem 5.8 in \cite{KaplralovTY21}, except that we have a prefix-graph $G$ and we only assign new weights to the newly arrived hyperedge $e$.

First, we note that $G$ is never disconnected. Otherwise, we need to set $\lambda = z_{uv}$ for some bridge $(u,v) \in e$ so that it has zero weight after the weight shift. 
By \lemref{lem:ub:ratio}, we have $z_{uv} = \frac{1}{q_{uv}} > \frac{\gamma - 1}{2\gamma\cdot q_{uv}}$. 
So, we would set $\lambda$ to $\frac{\gamma - 1}{2\gamma\cdot q_{uv}}$ instead, which is a contradiction. 
Therefore, the ST-potential $\Psi(G)$ is always well-defined.

From \lemref{lem:reduction}, whenever we make a weight shift in an unbalanced pair, the ST-potential of $G$ increases by
\[\log \left(1+\lambda \gamma \cdot q_{uv}-\lambda \cdot q_{uv}-\lambda^2 \gamma \cdot q_{uv}^2\right). \]
Next, we classify our weight shift operation into two types. When $\lambda = \frac{\gamma - 1}{2\gamma\cdot q_{uv}}$, the increment is at least $c_{\gamma} := \log \frac{1+(\gamma-1)^2}{4\gamma} > 0$. When $\lambda = z_{uv}$, the increment is positive, and $z_{u'v'}$ is set to zero. We specify these two cases as follows:
\begin{itemize}
    \item 1. $\Psi(G)$ at least increases by a constant $c_{\gamma} > 0$.
    \item 2. $\Psi(G)$ increases by a positive amount, and $z_{u'v'}$ is set to zero.
\end{itemize}
Now, we define $G_0$ to be the weighted graph at the initialization stage of \algref{alg:online:bw}, and we define $G_{\infty}$ to be a complete graph with node set $V$ with uniform weight $w(e)+ \sum_{(u,v) \in F} z_{uv}$, where $F$ is the edge set before adding hyperedge $e$. Note that $\sum_{u,v \in e} z_{uv} = w(e)$ by definition, then by the monotonicity of $\Psi$, $\Psi(G) \le \Psi(G_{\infty})$ for all $G$ obtained in \algref{alg:online:bw}. Therefore, there can be at most $\frac{\Psi(G_{\infty}) - \Psi(G_0)}{c_{\gamma}}$ steps of weight shifts of the first type.

For the second type, if $z_{uv} \neq 0$, then the number of zeros in the weight assignment for $e$ increases by $1$, which happens only a finite number of times. Therefore, we only need to consider $z_{uv} = 0$. In this case, we switch the weight of $(u,v)$ and $(u', v')$, and set $\cup_{u,v \in e} z_{uv}$ remains the same. Then, the weight assignment $z$ can only be in a finite number of stages without reverse. Thus, the number of weight shifts of type two is finite and hence \algref{alg:online:bw} terminates in finite steps.
\end{proof}

Due to the algorithm's construction, it must output an online $(\gamma, e)$-balanced weight assignment when it terminates. Therefore, we show the existence of such weight assignment. 

\begin{theorem}
Given a weighted graph $G=(V, F, z)$ and a hyperedge $e \subset V$ with weight $w(e)$, there exists an online $(\gamma, e)$-balanced weight assignment.
\end{theorem}

\section{Fast Online Hypergraph Spectral Sparsifier}
\seclab{sec:fast}

In this section, we give an online algorithm with a faster update time based on \cite{BansalST19}. 
We also include a result with success probability $1-\frac{\delta}{\poly(m)}$, which is used to give a streaming algorithm later, and show that the dependence on $\delta$ only increases the space by a $\log\log \frac{1}{\delta}$ factor.

As we mentioned previously, \cite{BansalST19} constructs the associated graph by assigning weight $w(e)$ to each edge of the hyperedge $e$, and they sample $e$ with probability
\[p_e \propto \max_{u,v \in e} r_{uv}, ~~ \mathrm{where~~} r_{uv} = w(e) \cdot (\chi_u-\chi_v)^\top\cdot L_G\cdot(\chi_u-\chi_v).\]
We extend this procedure to the online setting by maintaining a $2$-spectral approximation $\widehat{L_G(t)}$ to the graph Laplacian $L_G(t)$ at all times using online row sampling \cite{CohenMP20}. 
Then, we use $\widehat{L_G(t)}$ to define the sampling probabilities.
Our algorithm is displayed in \algref{alg:online:hg:sparsifier}.

\begin{algorithm}[!htb]
\caption{Online Hyperedge Spectral Sparsifier}
\alglab{alg:online:hg:sparsifier}
\begin{algorithmic}[1]
\State{{\bf Require:} Stream of $m$ hyperedges for hypergraph $H$ with rank $r$}
\State{{\bf Ensure:} Spectral sparsifier $\widehat{H}$ for $H$}
\State{$\widehat{H}\gets\emptyset$, $\rho\gets\O{\frac{r^4}{\eps^2}\log \frac{m}{\delta}}$}
\State{Let $G$ be the associated graph of $H$. Let $\frac{1}{2}\cdot L_G(t)\preceq\widehat{L_G(t)}\preceq 2\cdot L_G(t)$ for all $t\in[m]$} \Comment{Use online row sampling in \thmref{thm:ol:row}}
\For{hyperedge $e_t$} 
\For{$u,v\in e_t$}
\State{$\widehat{r_{u,v}}\gets w(e_t) \cdot (\chi_u-\chi_v)^\top\cdot\widehat{L_G(t)}^{-1}\cdot(\chi_u-\chi_v)$}
\EndFor
\State{$p_{e_t}\gets \min \{1, \rho\cdot\max_{u,v\in e_t}\widehat{r_{u,v}} \}$}
\State{With probability $p_{e_t}$, $\widehat{H}\gets\widehat{H}\cup\frac{1}{p_{e_t}}\cdot e_t$}
\EndFor
\end{algorithmic}
\end{algorithm}

The energy of the hypergraph reported by \algref{alg:online:hg:sparsifier} can be written as a random variable $\sum_{e\in E} X_e Q_e(x)$, where $X_e$ is $1/p_e$ with probability $p_e$ and $0$ otherwise. Then, the error of our approximation is $\sum_{e\in E} (X_e - 1)\cdot Q_e(x)$. \cite{BansalST19} simplifies this term to a sub-Gaussian random process $V_x$. 
Then, they bound its increment by a simpler Gaussian process $U_x$, which can be further bounded by Talagrand's chaining theorems. We show that their bound on the supremum of $V_x$ can be directly applied to prove our results. 
Intuitively, this is because the approximation $\widehat{L_G(t)}$ always gives an overestimate of the effective resistance, which means that we over-sample the hyperedges.

We remark that to show that the output $\widehat{H}$ of \algref{alg:online:hg:sparsifier} satisfies the spectral sparsification guarantee in Eq.~\eqref{eq:hg:sparsifier}, it suffices to show its correctness for hypergraphs where all hyperedges have size between $[r/2,r]$ (see Lemma 5.2 of \cite{BansalST19}). In addition, we have (see Lemma 5.5 of \cite{BansalST19})
\[\frac{2}{r(r-1)}x^\top L_Gx \le Q_H(x) \le \frac{4}{r}x^\top L_Gx ~~~ \mathrm{for~all}~ x\in \mathbb{R}^n,\]
if all hyperedges have size between $[r/2,r]$.
Thus, we only need to show that
\begin{align*} 
|Q_{\widehat{H}}(x) - Q_H(x)| \le \frac{\eps}{r^2} x^\top L_G x ~~~ \mathrm{for~all}~ x\in \mathbb{R}^n.
\end{align*}

Next, we introduce a normalized version of $Q_e(x)$.

\begin{definition}[Normalized energy, implicitly defined in Section 5.2 of \cite{BansalST19}]
For a hyperedge $e$ in $H$, let $p_e$ denote its sampling probability, and let $X_e$ be a random variable that is $1/p_e$ with probability $p_e$ and $0$ otherwise. Setting $z = L_G^{1/2}x$ and $Y_{uv} = L_G^{-1/2} L_{uv}L_G^{-1/2}$ gives
\[Q_e(x) = \max_{u,v \in e} x^\top L_{uv} x = \max_{u,v \in e} z^\top Y_{uv}z\]
We define $W_e(z) = \max_{u,v \in e} z^\top Y_{uv}z$ and $W_H(z) = \sum_{e\in E} W_e(z)$. Let $\widehat{H}$ be constructed with sampling distribution $p_e$ and rescaled factor $X_e$, then $W_{\widehat{H}}(z) = \sum_{e\in E} X_eW_e(z)$.
\end{definition}

With the normalized definition, our desired equation becomes
\begin{align} \label{eq:hg:sparsifier:simplified}
|W_{\widehat{H}}(z) - W_H(z)| \le \frac{\eps}{r^2} ~~~~~~~ \mathrm{for~all}~ z\in B_2
\end{align}
where $B_2$ is the unit-$\ell_2$-ball in the subspace restricted to the image of $L_G$. 
We prove Eq.~\eqref{eq:hg:sparsifier:simplified} in the following statements. First, we draw an equivalence between $\|Y_{uv}\|$ and the effective resistance $r_{uv}$.

\begin{fact}
\factlab{fac:y:r}
Let $\|\cdot\|$ denote the spectral norm of a matrix, we have $\|Y_{uv}\| = r_{uv}$.
\end{fact}
\begin{proof}
Notice that $Y_{uv} = L_G^{-1/2} L_{uv}L_G^{-1/2}$, so we have
\begin{align*}
Y_{uv} = & ~ w(e) \cdot L_G^{-1/2} \cdot(\chi_u-\chi_v)^\top(\chi_u-\chi_v) \cdot L_G^{-1/2} \\
= & ~ w(e) \cdot \left( (\chi_u-\chi_v) \cdot L_G^{-1/2} \right)^\top \cdot (\chi_u-\chi_v) \cdot L_G^{-1/2},
\end{align*}
where $e$ is the weight of the corresponding hyperedge in $H$. Notice $Y_{uv}$ is a rank-$1$ matrix spanned by vector $(\chi_u-\chi_v) \cdot L_G^{-1/2}$, thus, we have
\begin{align*}
\|Y_{uv}\| = & ~ \lambda_{\max}( Y_{uv}) = w(e) \cdot (\chi_u-\chi_v) \cdot L_G^{-1/2} \cdot \left( (\chi_u-\chi_v) \cdot L_G^{-1/2} \right)^\top   \\
= & ~ w(e) \cdot (\chi_u-\chi_v) \cdot L_G^{-1} \cdot (\chi_u-\chi_v)^\top = r_{uv}.
\end{align*}
\end{proof}

Now, we state the bound on the supremum of the random process $V_z$ defined in \cite{BansalST19}.
\begin{theorem}[Supremum of random process, see Theorem 5.15 in \cite{BansalST19}]
\thmlab{thm:supremum:vz}
For a hyperedge $e$ in $H$, let $G$ be the associated graph of $H$, we define the effective resistance of $e$ as $r_e = \max_{u,v \in e} r_{uv}$, where $r_{uv}$ is measured in graph $G$. Let $S \subset E(H)$ be a subset of hyperedges such that $\|Y_{uv}\| \le b$ for all $u,v \in e$ and $e\in S$, where $b$ is some constant. For independent Rademacher variables $\eps_e$ and vector $z \in \mathbb{R}^n$, let
\[V_z = \sum_{e \in S} \eps_e W_e(Z).\]
Then, we have $\Ex{\sup_{z \in B_2} V_z} = \O{\sqrt{b\log n}}$, and for all $u \ge 0$, we have
\[\PPr{\sup_{z \in B_2} V_z \ge \O{\sqrt{b \log n} + 2u \sqrt{b}}} \le 2e^{-u^2}.\]
\end{theorem}

With the above lemmas, we show the correctness of our online algorithm.
\begin{lemma}
\lemlab{lem:ol:cor}
Let $\widehat{H}$ be the output of \algref{alg:online:hg:sparsifier}, it is a $(1+\eps)$-multiplicative spectral sparsifier for $H$ with probability $1-\frac{\delta}{\poly(m)}$.
\end{lemma}
\begin{proof}
Fix a time $t$ in the stream, let $H$ be the hypergraph, and let $\widehat{H}$ be the sparsifier. Notice that the sampling probability of each hyperedge is solely determined by the calculation of the online leverage scores $\tau^\mathsf{OL}$, which is independent of the hyperedges sampled from previous arrivals. Therefore, we can view the sampling procedure as a re-ordered procedure. 

We state an iterative sampling process introduced by \cite{BansalST19}. Let $\tau^\mathsf{OL}_e = \max_{u,v\in e} \tau^\mathsf{OL}_{uv}$, where $\tau^\mathsf{OL}_{uv}$ is calculated at the time that $e$ is sampled. We round each sampling probability $p_e$ up to the nearest integer powers of $2$. This ensures $p_e \ge \min\{1, \tau_e^\mathsf{OL} \cdot \rho\}$, while at most doubling the expected sample complexity. 
Notice that hyperedges with $p_e=1$ do not contribute to the sampling error, so we can assume $\tau_e^\mathsf{OL} \cdot \rho < 1$ for all $e \in E(H)$. In addition, since we assume that all hyperedge weights are integers from $[\poly(n)]$, all eigenvalues in $L_G(t)$ are at most $\poly(m)$, and so the sampling probability $p$ is at least $\frac{1}{\poly(m)}$ for each hyperedge.
Thus, let $C_j = \{ e \in E(H) ~|~ p_e = 2^{-j}\}$ for each $j \in [\ell]$, and we have at most $\ell = \O{\log m}$ classes

Now, we view the process of sampling in the following way. Let $H_0 = H$, let $l = \O{\log m}$, and for $i \in [l]$, $H_i$ is obtained from $H_{i-1}$ by sampling each hyperedge $e$ from the set 
$\bigcup_{j \in \{l-i+1, l-i,\ldots,l\}} C_j$ independently with probability $1/2$ and doubling the weight of $e$ if it is sampled. Thus, an edge $e \in C_j$ that survives in $H_l$ is sampled with probability $p_e = 2^{-j}$ and has weight $X_e = 1/p_e$.

Now, for each $i\in[l]$, we define
\[W_{H_i}(z)=\sum_{j=0}^{l} \sum_{e \in C_j \cap E\left(H_i\right)} \max \left(1,2^{i+j-l}\right) W_e(z).\]
Since $H_0 = H$ and $H_l = \widehat{H}$, by triangle inequality we have
\[\left|W_{\widehat{H}}(z)-W_H(z)\right| \leq \sum_{i=1}^{l}\left|W_{H_i}(z)-W_{H_{i-1}}(z)\right|.\]
Taking the supremum over all $z$ in $B_2$ gives
\[\sup_{z\in B_2}\left|W_{\widehat{H}}(z)-W_H(z)\right| \leq \sum_{i=1}^{l} \sup_{z\in B_2}\left|W_{H_i}(z)-W_{H_{i-1}}(z)\right|.\]
From our definition of the iterative sampling procedure, we have
\[W_{H_i}(z)-W_{H_{i-1}}(z)=\sum_{j=\ell-i+1}^{l} \sum_{e \in C_j \cap E\left(H_{i-1}\right)} \eps_e 2^{i+j-\ell-1} W_e(z),\]
where $\eps_e$'s are independent Rademacher variables. Recall that we define $W_e(z) = \max_{u,v \in e} z^\top Y_{uv}z$. Then, for any $e \in \bigcup_{j \in \{l-i+1, l-i,\ldots,l\}} C_j$, we have
\[\|Y_{uv}\| =r_e = \max_{u,v \in e} \tau_{uv} \le \tau^\mathsf{OL}_e \le 2^{-j}/ \rho\]
where the first step follows from \factref{fac:y:r}, the second step follows from the equivalence between effective resistance and leverage score (see \lemref{lem:er:lev:score}), and the third step follows from the oversampling property of online leverage score (see \lemref{lem:lev:score:monotone}). So, $\|2^{i+j-1} Y_{uv}\| \le 2^{i-l} / \rho$. Applying \thmref{thm:supremum:vz} with $V_z = W_{H_i}(z) - W_{H_{i-1}}(z)$ and $u=\sqrt{\log \frac{m}{\delta}}$ gives 
\[\PPr{\sup_{z \in B_2} V_z \ge \O{\sqrt{2^{i-l}/\rho \cdot\log \frac{m}{\delta}}}} \le \frac{\delta}{\poly(m)}.\]
Recall that we set $\rho = \frac{r^4}{\eps^2}\log \frac{m}{\delta}$. Taking a union bound over the $l = \O{\log m}$ groups, we have
\[\sup_{z \in B_2} |W_{\widehat{H}}(z)-W_H(z)| \le \O{\sum_{i=1}^l \sqrt{2^{i-l}/\rho \cdot\log \frac{m}{\delta}}} = \O{\frac{\eps}{r^2}},\]
with probability $1-\frac{\delta}{\poly(m)}$. Taking a union bound over $m$ arrivals in the stream, the same bound still holds with probability $\frac{\delta}{\poly(m)}$. Thus, we show the correctness of our algorithm.
\end{proof}

The next statement bounds the sample complexity, which mainly follows by the upper bound on the sum of online leverage scores.
\begin{lemma}
\lemlab{lem:ol:sample}
With probability $1- \frac{1}{\poly(m)}$, \algref{alg:online:hg:sparsifier} samples $\O{\frac{nr^4}{\eps^2} \log n \log \frac{m}{\delta}}$ hyperedges. In addition, it uses $\O{n \log^2 n \log m}$ bits of working memory and $\poly(n)$ update time.
\end{lemma}
\begin{proof}
It suffices to upper bound the expected number of samples $\sum_{e \in E(H)} p_e$. By our definition of $p_e$, it is upper bounded by
\begin{align*}
\sum_{e \in E(H)} \rho \cdot \max_{u,v \in e} \tau_{uv}^\mathsf{OL} \le \rho \cdot \sum_{u,v \in E(G)} \tau^\mathsf{OL}(\ba_{uv}),
\end{align*}
where $\bA$ is the incidence matrix of the associated graph $G$. By \thmref{thm:sum:ol}, the sum of the online leverage score is bounded by $\O{n\log\kappa}$, where $\kappa=\|\bA\|_2\cdot\max_{i\in[n]}\|\bA_i^{-1}\|_2$ is the online condition number of $\bA$. Therefore, the expected number of samples is $\O{\rho n\log \kappa} = \O{\frac{nr^4}{\eps^2} \log \kappa \log \frac{m}{\delta}}$. Note that the inner randomness of sampling each hyperedge with probability $p_e$ is independent, then by standard concentration inequalities, the number of sampled hyperedges is $\O{\frac{nr^4}{\eps^2} \log \kappa \log \frac{m}{\delta}}$ with probability $1- \frac{1}{\poly(m)}$. Suppose that the condition of \factref{fac:kappa} is satisfied, then we have $\log \kappa = \O{\log n}$.

In addition, by the guarantee of online row sampling in \thmref{thm:ol:row}, it suffices to sample $\O{n \log n \log m}$ rows in the incidence matrix $\bA$ to construct the $2$-approximation $\widehat{L_G(t)}$ to the graph Laplacian $L_G(t) = \bA^\top \bA$. Since we only need to store $\widehat{L_G(t)}$ to compute the sampling probabilities, \algref{alg:online:hg:sparsifier} uses $\O{n \log^2 n \log m}$ bits of working memory. 
Moreover, the calculation of the sampling probabilities only requires $\poly(n)$ time.
\end{proof}

With \lemref{lem:ol:cor} and \lemref{lem:ol:sample}, we have the following result for the online hypergraph spectral sparsifier.
\begin{theorem}
[Online hypergraph spectral sparsifier]
\thmlab{thm:hyper:online}
Given a hypergraph $H=(V,E,w)$ with $n$ vertices and rank $r$, there exists an online algorithm that constructs a $(1+\eps)$-spectral sparsifier with probability $1-\frac{\delta}{\poly(m)}$ by sampling $\O{\frac{nr^4}{\eps^2}\log n\log \frac{m}{\delta}}$ hyperedges, using $\O{nr \log^2 n \log m}$ bits of working memory and $\poly(n)$ update time. 
\end{theorem}

Now, we specify the number of bits needed to store each sampled hyperedge.

\begin{remark} [Bits of precision]
\remlab{rem:bits:hyperedge}
Suppose that all hyperedge weights are integers from $[\poly(n)]$, we need $\O{r \log n + \log \log m}$ bits to store each sampled hyperedge.
\end{remark}
\begin{proof}
First, we note that it requires $r \log n$ bits to store the nodes included in the hyperedge.
Second, for each reweighted hyperedge sampled by our algorithm, we round its weight $w(e)$ to the nearest power of $(1+\O{\eps})$.
Recall that the energy of a vector $x$ in hypergraph $H$ is an additive function: $Q_H(x) = \sum_{e \in E} w(e) \cdot \max_{u,v} (x_u-x_v)^2$, therefore perturbing each $w(e)$ by a $(1+\O{\eps})$-fraction only has an additional $(1+\O{\eps})$-multiplicative error in our approximation guarantee.
Note that we sample each hyperedge with probability $p$ proportional to $w(e_t) \cdot (\chi_u-\chi_v)^\top\cdot\widehat{L_G(t)}^{-1}\cdot(\chi_u-\chi_v)$, where $\widehat{L_G(t)}$ is a $2$-spectral sparsifier of the associated graph $G$ of the hypergraph $H$, and we rescale each sampled hyperedge by $\frac{1}{p}$.
Thus, assuming that all hyperedge weights are integers from $[\poly(n)]$, all eigenvalues in $L_G(t)$ are at most $\poly(m)$, and so the sampling probability $p$ is at least $\frac{1}{\poly(m)}$ for each hyperedge.
The rescaling factor $\frac{1}{p}$ is then at most $\poly(m)$, so there are $\O{\frac{\log m}{\eps}}$ choices of powers that we need to store.
Therefore, we need $\O{\log \log m}$ bits to store each sampled hyperedge, assuming that $\frac{1}{\eps} \le \polylog(m)$.
\end{proof}

\section{Streaming Model}
\seclab{sec:framework}
In this section, we provide streaming algorithms for sparsification problems with nearly optimal space. We start by introducing the well-known merge-and-reduce approach in achieving $(1+\eps)$-coresets in a data stream. 
An \emph{online $(1+\eps)$-coreset} for graph sparsification for a graph $G$ defined by a stream of edges $e_1,\ldots,e_m$ is a subset $\widehat{G}$ of weighted edges of $G$ such that for any $x\in\mathbb{R}^n$ and any $t\in[m]$, we have $(1-\eps) Q_{G_t}(x) \le Q_{\widehat{G_t}}(x) \le(1+\eps)Q_{G_t}(x)$,
where $G_t$ is the set of hyperedges of $G$ that have arrived by time $t$.

Let $S(\cdot) \to \mathbb{R}_{>0}$ denote a mapping from input parameters to the sample complexity of an online algorithm. The online coresets for graph sparsification sample $S(n,\log m,\eps,\delta)$ edges for an input stream of length $m$ on a graph with $n$ nodes, accuracy $\eps\in(0,1)$, and failure probability $\delta$, with high probability. 
The merge-and-reduce approach partitions the stream into blocks of size $S\left(n,\log m,\frac{\eps}{2\log(mn)},\frac{\delta}{\poly(mn)}\right)$ and builds a $\left(1+\frac{\eps}{2\log(mn)}\right)$-coreset for each block, so that each coreset can be interpreted as the leaves of a binary tree with height at most $\log(mn)$, as the binary tree is balanced and has at most $m$ leaves corresponding to the edges that arrive in the data stream. 
For each node in the binary tree, a coreset of size $S\left(n,\log m,\frac{\eps}{2\log(mn)},\frac{\delta}{\poly(mn)}\right)$ is built from the coresets representing the two children of the node. 
Assuming that the coreset construction admits a merging procedure, i.e., by taking the graph consisting of the union of the weighted edges in each of the coresets, then the root of the tree represents a coreset for the entire stream with distortion at most $\left(1+\frac{\eps}{2\log(mn)}\right)^{\log(mn)}\le(1+\eps)$ and failure probability $\delta$.

\cite{Cohen-AddadWZ23} improves the above framework by adding an online sampling procedure ahead of the merge-and-reduce approach. Suppose that the online sampling procedure is nearly optimal. Then the input stream of the merge-and-reduce approach is significantly shorter.
In the graph sparsification problem, there is an online algorithm that samples $\O{\frac{n}{\eps^2}\log n}$ edges, so the coresets only have size $S\left(n,\log \log n,\frac{\eps}{2\log\log(n)},\frac{\delta}{\polylog(n)}\right)$. 
This turns $\polylog(n)$ factors into $\polylog \log(n)$ factors, which is more space-efficient for huge graphs. We summarize this framework in \algref{alg:framework}. 
A figure illustrating the merge-and-reduce approach is shown in \figref{fig:merge:reduce}.

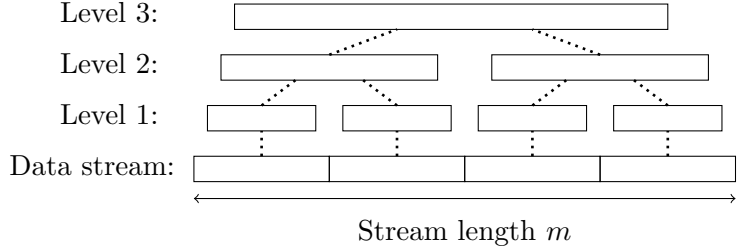
\begin{figure*}[!htb]
\centering
\begin{tikzpicture}[scale=0.45]
\draw[<->] (0,-0.5) -- (16,-0.5);
\node at (8,-1.5){Stream length $m$};

\node at (-3,0.5){Data stream:};
\draw (0,0) rectangle+(4,0.75);
\draw (4,0) rectangle+(4,0.75);
\draw (8,0) rectangle+(4,0.75);
\draw (12,0) rectangle+(4,0.75);
\draw[line width=1pt, dotted] (2,0.75) -- (2,0.75+0.75);
\draw[line width=1pt, dotted] (6,0.75) -- (6,0.75+0.75);
\draw[line width=1pt, dotted] (10,0.75) -- (10,0.75+0.75);
\draw[line width=1pt, dotted] (14,0.75) -- (14,0.75+0.75);

\node at (-2.5,2){Level 1:};
\draw (0.4,0.75+0.75) rectangle+(3.2,0.75);
\draw (4+0.4,0.75+0.75) rectangle+(3.2,0.75);
\draw (8+0.4,0.75+0.75) rectangle+(3.2,0.75);
\draw (12+0.4,0.75+0.75) rectangle+(3.2,0.75);
\draw[line width=1pt, dotted] (2,2.25) -- (3,3);
\draw[line width=1pt, dotted] (6,2.25) -- (5,3);
\draw[line width=1pt, dotted] (10,2.25) -- (11,3);
\draw[line width=1pt, dotted] (14,2.25) -- (13,3);

\node at (-2.5,3.5){Level 2:};
\draw (0.8,3) rectangle+(6.4,0.75);
\draw (8+0.8,3) rectangle+(6.4,0.75);
\draw[line width=1pt, dotted] (4,3+0.75) -- (6,4.5);
\draw[line width=1pt, dotted] (12,3+0.75) -- (10,4.5);

\node at (-2.5,5){Level 3:};
\draw (1.2,4.5) rectangle+(12.8,0.75);
\end{tikzpicture}
\caption{Merge and reduce framework on a stream of length $m$. The coresets at level 1 are precisely the hyperedges in the block, while the coresets at level $\ell>1$ are $\left(1+\O{\frac{\eps}{2\log(mn)}}\right)$-coresets of the hyperedges contained in the coresets of the children nodes in level $\ell-1$.}
\figlab{fig:merge:reduce}
\end{figure*}

\begin{algorithm}[!htb]
\caption{Streaming framework of \cite{Cohen-AddadWZ23}, using online sampling and merge-and-reduce}
\alglab{alg:framework}
\begin{algorithmic}[1]
\State{{\bf Require:} Stream $\calS$, online sampling procedure for $\calS$, merge-and-reduce procedure}
\State{{\bf Ensure:} Coreset on $\calS$}
\For{each update $s_t$ in the stream $\calS$}
\If{$s_t$ is sampled by online sampling}
\State{Add the corresponding update to $\calS'$}
\EndIf
\State{Run merge-and-reduce on $\calS'$}
\EndFor
\end{algorithmic}
\end{algorithm}

\subsection{Graph Spectral Sparsifier}
\seclab{sec:graph:stream}
We show that the streaming framework produces space-optimal streaming algorithms for constructing graph and hypergraph sparsifiers. First, in the offline setting, an efficient way to construct $(1+\eps)$-coresets for graph sparsifier is given by \cite{BatsonSS14}, using only $\O{\frac{n}{\eps^2}}$ edges. We provide the formal statement below.
\begin{theorem}[Offline algorithm for graph spectral sparsifier]
\thmlab{thm:g:offline}
\cite{BatsonSS14}
Given a graph $G=(V,E,w)$ with $n$ vertices, there exists an algorithm that constructs a $(1+\eps)$-spectral sparsifier with probability $1-\frac{1}{\poly(n)}$ using $\O{\frac{n}{\eps^2}}$ edges in $\poly(n)$ time.
\end{theorem}

We then state the result for the online graph spectral sparsifier in \cite{CohenMP20}.
\begin{theorem}[Online algorithm for graph spectral sparsifier]
\thmlab{thm:graph:online}
\cite{CohenMP20} Given a graph $G =(V, E, w)$ with $n$ vertices and $m$ edges defined by an insertion-only stream, there exists an online algorithm that constructs a $(1+\eps)$-spectral sparsifier with probability $1-\frac{1}{\poly(n)}$. The algorithm samples $\calO\left(\frac{n}{\eps^2} \log ^2 n\right)$ edges and uses $\O{n \log ^2 n}$ words of working memory and $\poly(n)$ update time.
\end{theorem}

Using the above subroutines, we show that our streaming algorithm has an optimal space up to a $\poly(\log \log n)$ factor, assuming $m=\poly(n)$ in the graph.
\begin{theorem}[Streaming algorithm for graph spectral sparsifier]
\thmlab{thm:graph:streaming}
Given a graph $G=(V,E,w)$ with $n$ vertices and $m$ edges defined by an insertion-only stream, there exists an algorithm that constructs a $(1+\eps)$-spectral sparsifier with probability $1-\frac{1}{\poly(n)}$.
The algorithm stores $\frac{n}{\eps^2}\cdot \poly(\log \log n)$ edges, i.e., $\frac{n}{\eps^2}\cdot \log n \poly(\log \log n)$ bits, and has $\poly(n)$ update time. 
\end{theorem}
\begin{proof}
The proof relies on the decomposability of the coresets, i.e., if $\widehat{G_1}$ is an $\eps$-sparsifier of $G_1$ and $\widehat{G_2}$ is a $\eps'$-sparsifier of $G_2$, then $\widehat{G_1} \cup \widehat{G_2}$ is a $\eps'$-sparsifier of $G_1 \cup G_2$.
This is simply because the energy of a vector $x$ on the graph $G$ is the sum of the energy of all edges in $G$: $Q_G(x) = \sum_{e \in G} Q_e(x)$.
Then, the edges in the stream at level $l+1$ in the merge-and-reduce binary tree construct a $\eps'$-sparsifier of the edges in the stream at level $l$.
Thus, the accumulative error of the coreset on the roof of a tree of height $h$ is $\eps' h$.

In our online framework \algref{alg:framework}, the hyperedges at the bottom level are the output $\calS'$ of the online algorithm, which contains $|\calS'|=\O{\frac{n}{\eps'^2}\log^2 n}$ edges by \thmref{thm:graph:online}.
Thus, the height of the merge-and-reduce data structure is $h = \O{\log \frac{|\calS'|}{|C|}}$, where $|C|$ is the size of each coreset in the merge-and-reduce structure. 
We define $C$ by \thmref{thm:g:offline}, which selects $\Theta(\frac{n}{\eps'^2})$ hyperedges. Hence, $h = \O{\log \log n}$. 
Therefore, taking $\eps'=\O{\frac{\eps }{\log \log n}}$ achieves an accumulative error of at most $\O{\eps' h = \eps}$, which implies a $(1+\eps)$-sparsifier. 
Then, $|C| = \O{\frac{n}{\eps'^2}} = \O{\frac{n}{\eps^2} \log \log^2 n}$. 
So, we need at most $\O{|C| \cdot h} = \O{\frac{n}{\eps^2} \log \log^3 n}$ words of memory in total.
The update time of our streaming algorithm directly follows from the results from \thmref{thm:graph:online} and \thmref{thm:g:offline}.

Last, we remark that the online row sampling procedure in \thmref{thm:graph:online} requires a sketch matrix of $\O{n \log^2 n}$ rows to define the online sampling probabilities, which exceeds our desired space complexity.
We note that it suffices to use the matrix consisting of the edges in the roof of our merge-and-reduce tree structure, i.e., the output of our streaming algorithm.
This gives the same guarantees for correctness and complexity as in \thmref{thm:graph:online} since our output is always a $2$-spectral approximation at each step $i$.
We defer the formal proof to \thmref{thm:working:memory}.
Then, since we only need to store our streaming output to define the online sampling probabilities, which has $\frac{n}{\eps^2}\cdot \poly(\log\log (n))$ edges, we use $\frac{n}{\eps^2}\cdot \log n \poly(\log\log (n))$ bits of working memory.
\end{proof}

\paragraph{Optimizing the Working Memory.}
We introduce the method to optimize the size of the sketch matrix.
In online row sampling, when a row $\ba_i$ arrives, the algorithm samples $\ba_i$ with probability proportional to the ``$\lambda$-ridge leverage score'':
\[\ell_i := \ba_{i-1}^T\left(\bA_{i-1}^T \bA_{i-1}+\lambda \mathbf{I}\right)^{-1} \ba_i,\]
where $\bA_{i-1}$ is the prefix matrix containing all the rows arrived before $\ba_i$ and $\lambda$ is some parameter defined in advance.
However, storing each $\bA_i$ explicitly requires a prohibitively large working memory.
Instead, it suffices to use a $2$-spectral approximation to $\bA_i$ to define the sampling probability, so \cite{CohenMP20} used $\widetilde{\bA_i}$, which is the approximation given by their online algorithm (see Algorithm 1 in \cite{CohenMP20}).  
However, $\widetilde{\bA_i}$ still has $\O{n \log^2 n}$ rows, which is prohibitively large for our purpose, so we replace it by the matrix consisting of the edges in the roof of our merge-and-reduce tree structure, i.e., the output of our streaming algorithm.
Next, we show that this modified online sampling scheme still gives the guaranties in \thmref{thm:graph:online}.
First, we state that oversampling gives a valid graph sparsifier, i.e., the online sampling probabilities only need to be defined by $\widetilde{\ell_i} \ge \ell_i$.

\begin{theorem}[Online Sampling Probability, \cite{CohenMP20}]
\thmlab{thm:online:scheme}
Let $\bA \in \mathbb{R}^{m \times n}$ be an incidence matrix of a (multi-)graph with $n$ vertices and $m$ edges, whose edge weights are integers in $[\poly(n)]$, let $\ba_i$ be the $i$-th row in $\bA$, and let $\bA_{i-1}$ be the prefix matrix containing all the rows arrived before $\ba_i$.
We define an online sampling scheme as follows.
Let $\widetilde{\ell}_i$ satisfies
\[\widetilde{\ell}_i \geq \mathbf{a}_i^T\left(\mathbf{A}_{i-1}^T \mathbf{A}_{i-1}+\lambda \mathbf{I}\right)^{-1} \mathbf{a}_i,\]
where $\lambda = \frac{\eps}{\poly(n)} \le \frac{\eps}{\sigma^2_{\min}(\bA)}$.
We set the sampling probability of row $\ba_i$ as $p_i = \min \{ c \widetilde{\ell}_i, 1\}$, where $c = \O{\frac{\log m}{\eps^2}}$, and if $\ba_i$ is sampled, we store the re-weighted row $\frac{\ba_i}{\sqrt{p_i}}$.
Then, let $\widetilde{\bA}_i$ be the matrix outputted by our sampling scheme  for each $i \in [m]$, with probability at least $\frac{1}{\poly(m)}$ we have
\[(1-\eps) \bA_i^\top \bA_i \preceq \widetilde{\bA}_i^T \widetilde{\bA}_i \preceq(1+\eps) \bA_i^\top \bA_i.\]
\end{theorem}

We remark that we set $c = \O{\frac{\log m}{\eps^2}}$, instead of $c = \O{\frac{\log n}{\eps^2}}$ in Algorithm 1 in \cite{CohenMP20}, to achieve a high success probability $1-\frac{1}{\poly(m)}$, enabling us to condition on the correctness upon each arrival $\ba_i$.
Next, we state an upper bound on the sum of ridge leverage scores.

\begin{theorem}[Sum of ridge leverage scores, \cite{CohenMP20}]
\thmlab{thm:sum:ridge}
In the sampling schemed described in \thmref{thm:online:scheme}, we define
\[\ell_i = \mathbf{a}_i^T\left(\mathbf{A}_{i-1}^T \mathbf{A}_{i-1}+\lambda \mathbf{I}\right)^{-1} \mathbf{a}_i,\]
where $\lambda = \frac{\eps}{\poly(n)} \le \frac{\eps}{\sigma^2_{\min}(\bA)}$.
Then we have $\sum_{i \in [m]} \ell_i = \O{n \log n}$.
\end{theorem}

The upper bound in \thmref{thm:sum:ridge} gives an upper bound on the sum of sampling probabilities; however, we cannot apply standard concentration inequalities to bound the number of sampled rows.
This is because in the online subroutine, the sample probability of the current $\ba_i$ depends on whether the previous rows are sampled, which is still the case after we replace the sketch matrix in the online subroutine by our output from the merge-and-reduce process.
Thus, we need to apply the following Freedman's inequality, which does not assume the independence of each sample.
\begin{theorem}[Freedman's inequality, \cite{Freedman75}] 
\thmlab{thm:free}
Let $Y_0, Y_1, \ldots, Y_n$ be a scalar martingale with difference sequence $X_1, \ldots, X_n$, i.e., $Y_0=0$ and $Y_t=Y_{t-1}+X_t$ for all $t \in[n]$.
Let $\left|X_t\right| \le R$ for all $t \in[n]$ with high probability. 
We define the predictable quadratic variation process of the martingale by $w_k:=\sum_{t=1}^k \underset{t-1}{\mathbb{E}}\left[X_t^2\right]$, for $k \in[n]$. Then for all $\epsilon \geq 0$ and $\sigma^2>0$, and every $k \in[n]$,
\[\PPr{\max _{t \in[k]}\left|Y_t\right|>\epsilon \text { and } w_k \leq \sigma^2} \leq 2 \exp \left(-\frac{\epsilon^2 / 2}{\sigma^2+R \epsilon / 3}\right).\]
\end{theorem}

Now, we show that the modified online subroutine satisfies the guaranties in \thmref{thm:graph:online}.
Thus, we only need to store a sketch matrix of $\O{\frac{n}{\eps^2}\poly(\log\log(n)}$ rows to define the online sampling probabilities, which achieves the desired space complexity.
\begin{theorem}
\thmlab{thm:working:memory}
Let graph $G =(V, E, w)$ with $n$ vertices and $m$ edges be defined by an insertion-only stream, let $\bA$ be the incidence matrix of $G$, and let $\ba_i$ be the $i$-th row in $\bA$.
We define an online sampling scheme as follows.
Let $\widetilde{\ell}_i$ satisfies
\[\widetilde{\ell}_i = \O{1} \cdot \mathbf{a}_i^T\left(\widehat{\bA}_{i-1}^T \widehat{\bA}_{i-1}+\lambda \mathbf{I}\right)^{-1} \mathbf{a}_i,\]
where $\lambda = \frac{\eps}{\poly(n)} \le \frac{\eps}{\sigma^2_{\min}(\bA)}$ and $\widehat{\bA}_{i-1}$ is the output of our streaming algorithm the last round (see the description in \thmref{thm:graph:streaming}).
We set the sampling probability of row $\ba_i$ as $p_i = \min \{ c \widetilde{\ell}_i, 1\}$, where $c = \O{\frac{\log m}{\eps^2}}$, and if $\ba_i$ is sampled, we store the re-weighted row $\frac{\ba_i}{\sqrt{p_i}}$.
Then, let $\widetilde{\bA}_i$ be the matrix outputted by our sampling scheme  for each $i \in [m]$, with probability at least $\frac{1}{\poly(m)}$ we have
\[(1-\eps) \bA_i^\top \bA_i \preceq \widetilde{\bA}_i^T \widetilde{\bA}_i \preceq(1+\eps) \bA_i^\top \bA_i.\]
The algorithm samples $\calO\left(\frac{n}{\eps^2} \log m \log n\right)$ edges and uses $\poly(n)$ update time.
\end{theorem}
\begin{proof}
First, we prove the correctness of our algorithm by applying \thmref{thm:sum:ridge} in an iterative way.
For initial stage, since $\ba_1$ is sampled by the online subroutine, $\widetilde{\bA}_1$ is a $(1+\eps)$-spectral approximation of $\bA_1$, and $\widehat{\bA}_1$ is a $2$-spectral approximation of $\bA_1$.
Now, conditioning on that $\widehat{\bA}_{i-1}$ is a $2$-spectral approximation of $\bA_{i-1}$.
Thus, the eigenvalues of $\widehat{\bA}_{i-1}^T \widehat{\bA}_{i-1}+\lambda \mathbf{I}$ and $\bA_{i-1}^T \bA_{i-1}+\lambda \mathbf{I}$ are within a constant fraction.
Then, by our construction, we have
\[\widetilde{\ell}_i = \O{1} \cdot \mathbf{a}_i^T\left(\widehat{\bA}_{i-1}^T \widehat{\bA}_{i-1}+\lambda \mathbf{I}\right)^{-1} \mathbf{a}_i \ge \mathbf{a}_i^T\left(\bA_{i-1}^T\bA_{i-1}+\lambda \mathbf{I}\right)^{-1} \mathbf{a}_i.\]
Now, consider the next arrival $\ba_i$, by the guarantee of \thmref{thm:sum:ridge}, the output $\widetilde{\bA}_i$ of the online procedure is still a $(1+\eps)$-approximation with probability $1 - \frac{1}{\poly(m)}$.
Condition on this event, by the analysis in \thmref{thm:graph:streaming}, the output $\widehat{\bA}_i$ of our streaming algorithm is a $2$-approximation with probability $1 - \frac{1}{\poly(m)}$.
Following this induction, we prove that $\widetilde{\bA}_i$ is a $(1+\eps)$-approximation for each $i \in [m]$.
Since for each step, the success probability is at least $1 - \frac{1}{\poly(m)}$, we can union bound across the $m$ arrivals.

Next, we bound the number of sampled rows.
From our induction, we also prove that $\widehat{\bA}_i$ is a $(1+\eps)$-approximation for each $i \in [m]$ with probability $1 - \frac{1}{\poly(m)}$. 
By our construction of the scores $\widetilde{\ell}_i$, we have
\[\widetilde{\ell}_i = \O{1} \cdot \mathbf{a}_i^T\left(\widehat{\bA}_{i-1}^T \widehat{\bA}_{i-1}+\lambda \mathbf{I}\right)^{-1} \mathbf{a}_i = \O{1} \cdot  \mathbf{a}_i^T\left(\bA_{i-1}^T\bA_{i-1}+\lambda \mathbf{I}\right)^{-1} \mathbf{a}_i.\]
Then, by the upper bound in \thmref{thm:sum:ridge}, we have $\sum_{i \in [m]} \widetilde{\ell}_i = \O{1}\cdot \sum_{i \in [m]} \widetilde{\ell}_i = \O{n \log n}$.
Moreover, we have $\sum_{i \in [m]} p_i \le c \cdot \sum_{i \in [m]} \widetilde{\ell}_i = \O{\frac{n}{\eps^2}\log m \log n}$.

Next, for each row $\ba_i$, we define $X_i$ as the indicator variable for $\ba_i$ shifted by $p_i$, i.e., $X_i = 1-p_i$ if $\ba_i$ is sampled and $X_i = -p_i$ otherwise.
Thus, $\underset{i-1}{\mathbb{E}}\left[X_i\right] =0$.
Then we define the martingale as $Y_0 =0$ and $Y_i = \sum_{j \in [i]} X_j$ for each $i \in [m]$, which is the difference between the number of sampled rows and $\sum_{i \in [m]} p_i$.
For each $i \in [m]$, we have $|X_i| \le 1$ and $\underset{i-1}{\mathbb{E}}\left[X_i^2\right] = p_i(1-p_i) \le p_i$, so $w_i := \sum_{j \in [i]} \underset{j-1}{\mathbb{E}}\left[X_j^2\right] \le \sum_{j \in [i]} p_j$.
Now, applying the Freedman's inequality (c.f. \thmref{thm:free}) with $\epsilon = \log m \cdot \sqrt{\sum_{i \in [m]} p_i}$ and $\sigma^2 = \sum_{i \in [m]} p_i$, we have
\[\PPr{\left|Y_m\right|> \log m \cdot \sqrt{\sum_{i \in [m]} p_i}} \leq \frac{1}{\poly(m)}.\]
That is, the number of sampled rows is $\O{\frac{n}{\eps^2}\log m \log n}$ with high probability.
\end{proof}

\subsection{Optimizing the Working Memory of Online Hypergraph Sparsification}
In this section, we optimize the working memory required in our online hypergraph sparsifier, which is applied as a black-box in our streaming algorithm for hypergraph sparsification in \algref{alg:framework}.
Recall that both algorithms in \thmref{thm:hg:online} and \thmref{thm:hyper:online} store a sketch matrix obtained by online row sampling the associated graph, which is used to define the sampling probabilities of the hyperedges.
We show that the sketch matrix can be replaced by the output of the streaming algorithm introduced in \secref{sec:graph:stream} with fewer rows.
The construction is based on the subroutine in \thmref{thm:working:memory}.
The following is the improved statement of \thmref{thm:hg:online}.

\begin{theorem}
[Online hypergraph spectral sparsifier]
\thmlab{thm:hg:online:opt:memory}
Given a hypergraph $H=(V,E,w)$ with $n$ vertices, $m$ hyperedges and rank $r$, there exists an online algorithm with $n r\log n \cdot \poly(\log \log m)$ bits of working memory that constructs a $(1+\eps)$-spectral sparsifier with probability $1-\frac{1}{\poly(m)}$ by sampling $\O{\frac{n}{\eps^2}\log n \log m \log r}$ hyperedges. 
\end{theorem}
\begin{proof}
Recall that we maintain a matrix $\bM$ that is a $2$-approximate spectral approximation of the re-weighted incidence matrix $\bZ_t^{1/2} \bA_t$ in \algref{alg:online:hg:sparsifier:opt}. 
We slightly change the procedure to define $\bM$: when a hyperedge $e_{t+1}$ arrives, we run $\textsc{GetWeightAssignment}(\bM, e)$ to obtain the weight assignment vector $z_{t+1}$; Then, for each edge $uv$ in the clique of $e_{t+1}$, we sample $\ba_{uv} \cdot \sqrt{z_{t+1,uv}}$ by online row sampling and push it to the merge-and-reduce data structure if sampled (see \figref{fig:merge:reduce}); then we set $\bM$ to be the weighted matrix obtained by the merge-and-reduce procedure.

By the guarantee of \thmref{thm:graph:streaming}, sampling $n \poly(\log \log m)$ rows to $\bM$ by the merge-and-reduce process still gives a $2$-approximation to $\bZ_t^{1/2} \bA_t$.
Moreover, we can optimize the working memory by the subroutine in \thmref{thm:working:memory}.
That is, we can use $\bM$ itself to define the online sampling probability for the next row in matrix $\bZ^{1/2} \bA$.

Then, it suffices to show that the updated procedure still satisfies the independence guarantees in the conditions of \thmref{thm:correctness}.
The analysis follows from the decoupling technique in \lemref{lem:indep}.
The above procedure that constructs the sketch matrix $\bM$, which includes the online row sampling process and the merge-and-reduce process, and the sampling of the hyperedges are separate procedures with independent inner randomness.
So, if we fix the inner randomness in the online row sampling process and the merge-and-reduce process in the construction of $\bM$, then the sampling probabilities are fixed.
Then, given these fixed sampling probabilities, the sampling of the hyperedges can be viewed as an offline procedure, and for each $t$, the sampling of $e_t$ is independent of whether the previous hyperedges are sampled.
\end{proof}

We next show the improved statement of \thmref{thm:hyper:online}.
\begin{theorem}
[Online hypergraph spectral sparsifier]
\thmlab{thm:hyper:online:opt:memory}
Given a hypergraph $H=(V,E,w)$ with $n$ vertices and rank $r$, there exists an online algorithm that constructs a $(1+\eps)$-spectral sparsifier with probability $1-\frac{\delta}{\poly(m)}$ by sampling $\O{\frac{nr^4}{\eps^2}\log n\log \frac{m}{\delta}}$ hyperedges, using $nr \log n \cdot \poly(\log \log m)$ bits of working memory and $\poly(n)$ update time. 
\end{theorem}
\begin{proof}
We use the same construction in \thmref{thm:hg:online:opt:memory}.
Recall that we do online row sampling on the incidence matrix $\bA$ and obtain a $2$-approximation $\widehat{L_G(t)}$ to the Laplacian $L_G(t) = \bA^T \bA$ of the associated graph.
Now, when an edge $\ba_i$ is sampled by the online row sampling procedure, we again push it to the merge-and-reduce structure, and use the output to define the next online sampling probability.
By \thmref{thm:graph:streaming} and \thmref{thm:working:memory}, the output is still a $2$-approximation to $\bA$, and hence it suffices for our analysis in \secref{sec:fast}.
\end{proof}

\subsection{Hypergraph Spectral Sparsifier}
For the hypergraph sparsification problem, we apply the offline algorithms introduced by \cite{JambulapatiLS23,Lee23} to merge the coresets in the two child nodes.
\begin{theorem}[Offline algorithm for hypergraph spectral sparsifier]
\thmlab{thm:hyper:offline}
\cite{JambulapatiLS23,Lee23}
Given a hypergraph $H=(V,E,w)$ with $n$ vertices, $m$ hyperedges and rank $r$, there exists an algorithm that constructs a $(1+\eps)$-spectral sparsifier with probability $1-\frac{1}{\poly(m)}$ using $\O{\frac{n}{\eps^2}\log n\log r}$ hyperedges in $\tO{mr}$ time.
\end{theorem}

Notice that there is no guarantee that we can find a $(\gamma,e)$-balanced weight assignment in \defref{def:ol;gamma:bal} in polynomial time, so using the online algorithm given by \thmref{thm:hg:online:opt:memory} as a subroutine does not imply a fast update time. 
Therefore, we instead apply \thmref{thm:hyper:online:opt:memory} with $\poly(n)$ update time as the online sampling subroutine. This will lose a $\poly(r)$ factor in the online sample complexity; however, it is acceptable since the streaming framework reduces it to $\polylog(r)$. 
We also include the result by applying \thmref{thm:hg:online:opt:memory} in our statement, which has a better space bound.

\begin{theorem}[Streaming algorithm for hypergraph spectral sparsifier]
\thmlab{thm:hyper:streaming}
Given a hypergraph $H=(V,E,w)$ with $n$ vertices, $m$ edges and rank $r$ defined by an insertion-only stream, there exists an algorithm that with probability $1-\frac{1}{\poly(m)}$, constructs a $(1+\eps)$-spectral sparsifier, storing $\frac{n}{\eps^2}\log n \cdot \poly(\log r, \log \log m)$ hyperedges, i.e., $\frac{rn}{\eps^2}\log^2 n \cdot \poly(\log r, \log \log m)$ bits, and using $\poly(n)$ update time. 
There is also an algorithm that stores $\frac{n}{\eps^2}\log n \log r\cdot \poly(\log \log m)$ hyperedges, i.e. $\frac{rn}{\eps^2}\log^2 n \log r\cdot \poly(\log \log m)$ bits, and uses exponential update time.
\end{theorem}
\begin{proof}
We have $|\calS'| = \O{\frac{nr^4}{\eps'^2}\log n\log m}$ by \thmref{thm:hyper:online:opt:memory} and $|C| = \Theta\left(\frac{n}{\eps'^2}\log n\log r\right)$ by \thmref{thm:hyper:offline}.
Thus, the height $h = \O{\log \frac{|\calS'|}{|C|}} = \O{\log r + \log \log m}$.
Taking $\eps' = \eps/h$, our total space usage is at most $\O{|C| \cdot h} = \O{\frac{n}{\eps'^2}\log n\log r \cdot (\log r + \log \log m)^3}$ words.

Each time a hyperedge arrives, the online algorithm needs $\poly(n)$ time to process by \thmref{thm:hyper:online:opt:memory}.
In addition, we at most need to merge $2h$ coresets each with $m' = |C|$ edges, which takes $\tO{m'r} \cdot h = \poly(n)$ update time by \thmref{thm:hyper:offline}.
Thus, we need $\poly(n)$ update time in total.

If we use \thmref{thm:hg:online:opt:memory} instead, then $|S'| = \O{\frac{n}{\eps^2}\log n \log m \log r}$ and $h = \O{\log \log m}$. 
Hence, we need $\O{|C| \cdot h} = \O{\frac{n}{\eps^2}\log n \log r \cdot \log \log^3 m}$ words of memory in total.

Notice that the failure probabilities of both the offline and the online subroutines are $\frac{1}{\poly(m)}$, thus, the failure probability of our streaming algorithm follows by a union bound across all times.
\end{proof}

\subsection{Hypergraph Spectral Sparsifier - High Probability}
\seclab{sec:high:prob}
In this section, we provide a streaming algorithm that succeeds with probability at least $1-\frac{\delta}{\poly(m)}$. 
In our online algorithm in \thmref{thm:hyper:online}, we lose a $\log \frac{1}{\delta}$ factor in the sample complexity to boost the failure probability to $\delta$, which may be prohibitively large if $\delta = \frac{1}{c^n}$. 
With the streaming framework of \cite{Cohen-AddadWZ23}, we reduce it to a $\log \log \frac{1}{\delta}$ factor.

\begin{theorem}[Small failure probability]
\thmlab{thm:high:prob:stream}
Given a hypergraph $H=(V,E,w)$ with $n$ vertices, $m$ hyperedges, and rank $r$ defined by an insertion-only stream, there exists an algorithm that constructs a $(1+\eps)$-spectral sparsifier with probability $1-\frac{\delta}{\poly(m)}$ storing $\frac{n}{\eps^2}\log n \cdot \poly(\log r, \log \log \frac{m}{\delta})$ hyperedges, i.e., $\frac{rn}{\eps^2}\log^2 n \cdot \poly(\log r, \log \log \frac{m}{\delta})$ bits.
\end{theorem}
\begin{proof}
First, we describe a deterministic offline algorithm that constructs a $(1+\eps)$-hypergraph spectral sparsifier. It loses more $\poly(n)$ factors at runtime while it does not fail. Since \thmref{thm:hyper:offline} gives a randomized algorithm that finds the sparsifier with $\O{\frac{n}{\eps^2} \log n \log r}$ hyperedges given any hypergraph with non-zero probability, there must exist such a sparsifier. Let $m$ denote the total number of hyperedges. We simply traverse through all possible groups of $\O{\frac{n}{\eps^2} \log n \log r}$ hyperedges, where there are $\binom{m}{\O{\frac{n}{\eps^2} \log n \log r}}$ of them. For each group, we test it on the net of points $x \in \mathbb{R}^n$ given by the chaining argument in \cite{JambulapatiLS23} and report this group of hyperedges if it successfully approximates the energy of the hypergraph $Q_H(x)$. 

Then, we use the online algorithm in \thmref{thm:hyper:online:opt:memory} with failure probability $\frac{\delta}{\poly(m)}$ to construct the stream $\calS'$ and the offline deterministic algorithm mentioned above to construct the coresets. 
The calculation of space complexity follows from the same argument in \thmref{thm:hyper:streaming}.
Here, we have $|\calS'| = \O{\frac{nr^4}{\eps'^2}\log n\log \frac{m}{\delta}}$ by \thmref{thm:hyper:online:opt:memory} and $|C| = \Theta\left(\frac{n}{\eps'^2}\log n\log r\right)$ by \thmref{thm:hyper:offline}.
Thus, the height $h = \O{\log \frac{|\calS'|}{|C|}} = \O{\log r + \log \log \frac{m}{\delta}}$.
Taking $\eps' = \eps/h$, our total space usage is at most $\O{|C| \cdot h} = \O{\frac{n}{\eps'^2}\log n\log r \cdot (\log r + \log \log \frac{m}{\delta})^3}$ words.
Unfortunately, we do not have the $\poly(n)$ update time guarantee due to the offline deterministic algorithm.

Notice that the failure probability of the online subroutines is $\frac{\delta}{\poly(m)}$, and the offline subroutine is deterministic, thus, the failure probability of our streaming algorithm follows by a union bound across all times.
\end{proof}

\subsection{Adversarially Robust Hypergraph Sparsification}

In this section, we apply the result in \secref{sec:high:prob} to achieve an adversarially robust streaming algorithm.
The adversarially robust model can be captured by the following two-player game between a streaming algorithm $\calP$ and an adversary $\calA$ that produces adaptive inputs to $\calP$. 
Given a query function $\calQ$, the game proceeds over $m$ rounds, and in the $t$-th round:
\begin{enumerate}
\item
$\calA$ determines an input $s_t$, which possibly depends on previous outputs from $\calP$.
\item
$\calP$ processes $s_t$ and outputs its answer $Z_t$ to the query function $\calQ$.
\item
$\calA$ receives and records the response $Z_t$.
\end{enumerate}
The goal of $\calP$ is to produce a correct answer $Z_t$ to the query function $\calQ$ based on the previously arrived data stream $\{s_1,\ldots,s_t\}$ sent by the adaptive adversary $\calA$, at all times $t \in [m]$.  
We now provide an adversarially robust streaming algorithm for hypergraph sparsification.

We begin with the definition of the $\eps$-flip number, which upper bounds the number of multiplicative increments of the output of a streaming algorithm.

\begin{definition}[$\eps$-flip number]
Let $\eps \geq 0$ and let $y=\left(y_0, y_1, \ldots, y_m\right)$ be a sequence of real numbers. Then, the $\eps$-flip number $\lambda_{\eps,m}(y)$ of the sequence $y$ is the maximum $k$ for which there exists $0 \leq i_1<\ldots<i_k \leq m$ so that $y_{i_{j-1}} \notin(1 \pm \eps) y_{i_j}$ for every $j=2,3, \ldots, k$.
In particular, for a function $g: \mathbb{R}^n \to \mathbb{R}$ and a class of data stream $\mathcal{S} \subset [n]^m$, the $(\eps, m)$-flip number $\lambda_{\eps, m}(g)$ of $g$ over $\mathcal{S}$ is the maximum $\eps$-flip number of the sequence $\bar{y}=\left(y_0, y_1, \ldots, y_m\right)$ defined by $y_t=g\left(s_1,\ldots, s_t\right)$, over all choices of data streams $S = (s_1,\ldots, s_m) \in \mathcal{S}$.
\end{definition}

\cite{BenEliezerJWY20} introduces a framework for vector-based problems, i.e., one computes a target function $g$ on the frequency vector induced by the data stream, which transforms any non-robust streaming algorithm to an adversarially robust streaming algorithm. 
The core idea is that we only change the estimate $\widehat{g}$ when it increases by an $\eps$-fraction, so if the $\eps$-flip number is small, then the total number of input streams that we need to handle is also relatively small, and we can union bound across them by setting a sufficiently small failure probability $\delta$ for the non-robust streaming algorithm.
We adapt their framework to solve the sparsification problems.

\paragraph{Robust Graph and Hypergraph Sparsification.}
We start by introducing an adversarially robust streaming algorithm for graph sparsification. 
A challenge in our problem is that our output is a sparsified graph that preserves the energy of all vectors, which is not a real number, so we need to define a proper way to change the output.
Recall that the energy of a vector $x$ on $G$ is the quadratic form $x^T L_G x$, where $L_G$ is the graph Laplacian.
Thus, we use the eigenvalues of $L_G$ to decide whether we change the output.
We state the formal algorithm and its guarantees as follows.

\begin{theorem}[Robust graph sparsification]
\thmlab{thm:robust:graph}
Given a (multi-)graph $G=(V,E,w)$ with $n$ vertices and $m$ edges defined by an insertion-only stream, there exists an adversarially robust algorithm that constructs a $(1+\eps)$-spectral sparsifier with probability $1-\frac{\delta}{\poly(m)}$ storing $\frac{n}{\eps^2} \cdot \poly(\log n, \log \log \frac{m}{\delta})$ edges, i.e., $\frac{n}{\eps^2} \cdot \poly(\log n, \log \log \frac{m}{\delta})$ bits.
\end{theorem}
\begin{proof}
First, we bound the $\eps$-flip number in our problem, which is the number of times that $L_{G'} \succeq (1+\eps) \cdot L_G$, where $G'$ is a graph later in the data stream. 
Recall that our stream is insertion-only, so all the eigenvalues of $G$ are increasing.
Note that the flip does not occur if none of the eigenvalues of $G$ increases by an $\eps$-fraction.
Here, without loss of generality, we assume that all eigenvalues are non-zero, since we can consider it a flip when the eigenvalue first becomes non-zero, and there are at most $n$ such flips.
Assuming that all edge weights are within $\poly(n)$, the eigenvalues at the end of the stream are upper bounded by $m \poly(n)$, and so each eigenvalue can increase by an $\eps$-fraction for at most $\frac{\log m}{\eps}$ times.
Thus, the $\eps$-flip number is at most $\lambda_{\eps,m}(G) = \frac{n \log m}{\eps}$.

We next state the formal algorithm for robust graph sparsification. We run a non-robust streaming algorithm with parameters $\eps' = \frac{\eps}{8}$ and $\delta' = \frac{\delta}{\poly(m) \cdot \binom{m}{\lambda} T^{\lambda}}$, where $\lambda := \lambda_{\frac{\eps}{8},m}(G)$ and $\log T$ are the bits of precision.
For the sequence of outputs $\widehat{G}_{1}, \ldots, \widehat{G}_{m}$, let $\calG_1 = \widehat{G}_{1}$, we set $\calG_{t} = \calG_{t-1}$ when all eigenvalues of $L_{\widehat{G}_{t}}$ are within a factor of $(1+\frac{\eps}{8})$ of that of $L_{\calG_{t-1}}$, otherwise we set $\calG_{t} = \widehat{G}_t$.
The output to the adversary is the sequence $\calG_1, \ldots \calG_m$.

We prove the correctness of the algorithm as follows.
Consider a fixed time $t$ when we change the output, and let $t'$ be any time after $t$ before the next output change.
Note that $\widehat{G}_{t'}$ is a $\frac{\eps}{8}$-sparsifier of $G_{t'}$ and all eigenvalues of $L_{\widehat{G}_{t'}}$ are within a factor of $(1+\frac{\eps}{8})$ of that of $L_{\widehat{G}_{t}}$, so all eigenvalues of $L_{G_{t'}}$ are within a factor of $(1+\frac{\eps}{2})$ of that of $L_{G_{t}}$.
This implies that $L_{G_{t'}} \preceq (1+\frac{\eps}{2}) \cdot L_{G_t}$, and so $x^\top L_{G_{t'}} x \le (1+\frac{\eps}{2}) \cdot  x^\top L_{G_{t}} x$ for all vectors $x \in \mathbb{R}^{n}$.
Therefore, $\calG_{t'} = \widehat{G}_{t}$ is an $\eps$-sparsifier of $G_{t'}$.
Moreover, we can assume the adversary to be deterministic (see the proof of Lemma 3.5 in \cite{BenEliezerJWY20}.
Then, the number of output sequences is at most $\binom{m}{\lambda} T^{\lambda}$ and they at most determine $\binom{m}{\lambda} T^{\lambda}$ choices of input streams.
Since we set the failure probability as $\delta' = \frac{\delta}{\poly(m) \cdot \binom{m}{\lambda} T^{\lambda}}$, we can union bound across all choices of input streams, ensuring the correctness of our algorithm.

By \thmref{thm:high:prob:stream}, we have that the non-robust streaming algorithm with parameters $\eps,\delta$ stores $\frac{n}{\eps^2}\log n \cdot \poly(\log \log \frac{m}{\delta})$ edges.
Therefore, since we require $\log T = \O{\log n + \log\log m}$ bits of precisions, the robust algorithm stores $\frac{n}{\eps'^2}\log n \cdot \poly( \log \log \frac{m}{\delta'}) = \frac{n}{\eps^2}\log^2 n \cdot \poly( \log \log \frac{m}{\delta})$ edges.
\end{proof}

Next, we apply the above result to construct a robust hypergraph sparsification algorithm.
Note that the energy of a vector $x$ on a hyperedge $e$ is defined as the maximum energy of $x$ on each edge $u,v$ in the clique of $e$, so we cannot directly define the $\eps$-flip number by the sparsified hypergraph.
Instead, we run a separated robust subroutine that constructs a sparsifier for the associated graph with higher accuracy, then we decide whether to change the output by the eigenvalues of the sparsified associated graph.
We introduce the formal algorithm and analyze its guarantees as follows.

\begin{theorem}[Robust hypergraph sparsification]
Given a graph $H=(V,E,w)$ with $n$ vertices, $m$ edges, and rank $r$ defined by an insertion-only stream, there exists an adversarially robust algorithm that constructs a $(1+\eps)$-spectral sparsifier with probability $1-\frac{\delta}{\poly(m)}$ storing $\frac{n}{\eps^2}\poly(\log n, \log r, \log \log \frac{m}{\delta})$ hyperedges and $\frac{nr^4}{\eps^2} \cdot \poly(\log n, \log \log \frac{m}{\delta})$ edges in the associated graph, i.e., $\frac{nr^5}{\eps^2} \cdot \poly(\log n, \log r, \log \log \frac{m}{\delta})$ bits in total.
\end{theorem}
\begin{proof}
We first state the formal algorithm for robust hypergraph sparsification. 
We run the robust streaming graph sparsification algorithm in \thmref{thm:robust:graph} on the associated graph $G$ of the hypergraph $H$ defined by the stream with parameters $\tilde{\eps} = \frac{\eps}{8r^2}$ and $\tilde{\delta} = \delta$, and we have the outputs $\calG_1, \ldots, \calG_m$.
Here, we note that $G$ is defined by the standard associated graph definition in \defref{def:associated}, but not the balanced-weight version introduced in \secref{sec:ol}.
We run a non-robust hypergraph sparsification algorithm separately with parameters $\eps' = \frac{\eps}{8}$ and $\delta' = \frac{\delta}{\poly(m) \cdot \binom{m}{\lambda} T^{\lambda}}$, where $\lambda := \lambda_{\frac{\eps}{8r},m}(G)$ and $\log T$ are the bits of precision, and we have the outputs $\widehat{H}_{1}, \ldots, \widehat{H}_{m}$.
Let $\calH_1 = \widehat{H}_{1}$, we set $\calH_{t} = \calH_{t-1}$ if $\calG_t = \calG_{t-1}$, otherwise we set $\calH_{t} = \widehat{H}_t$.
The output to the adversary is the sequence $\calH_1, \ldots \calH_m$.

We prove the correctness of the algorithm as follows.
Consider a fixed time $t$ when we change the output, and let $t'$ be any time after $t$ before the next output change.
We note that $t$ is also the time when the robust graph sparsification algorithm changes its output.
Then, from the analysis in \thmref{thm:robust:graph}, we have $L_{G_{t'}} \preceq (1+\frac{\eps}{2r^2}) \cdot L_{G_{t}}$ and $x^\top L_{G_{t'}} x \le (1+\frac{\eps}{2r^2}) \cdot  x^\top L_{G_{t}} x$ for all vectors $x \in \mathbb{R}^{n}$.
Recall that the energy of $x$ on the hypergraph $H$ is $Q_H(x) = \sum_{e\in E} \max_{u,v \in e} x^\top L_{uv} x $, which is at least $\frac{1}{r^2} \cdot \sum_{e\in E} \sum_{u,v \in e} x^\top L_{uv} x = \frac{1}{r^2}\cdot x^\top L_{G} x$, since there are at most $\binom{r}{2}$ edges $(u,v) \in e$.
Then, we have
\[Q_H(x) \le Q_{H'}(x) + \frac{\eps}{2r^2} \cdot  x^\top L_{G_{t}} x \le (1+ \frac{\eps}{2}) \cdot Q_{H'}(x).\]
Therefore, $\calH_{t'} = \widehat{H}_{t}$ is an $\eps$-sparsifier of $H_{t'}$.
Since we set the failure probability as $\delta' = \frac{\delta}{\poly(m) \cdot \binom{m}{\lambda} T^{\lambda}}$, we can union bound across all $\binom{m}{\lambda} T^{\lambda}$ choices of input streams, ensuring the correctness of our algorithm.

By \thmref{thm:high:prob:stream}, we have that the non-robust streaming algorithm with parameters $\eps',\delta'$ stores $\frac{n}{\eps^2}\log n \cdot \poly(\log r, \log \log \frac{m}{\delta})$ edges.
Note that the number of times that we change the output hypergraph is $\lambda = \O{\frac{nr \log(mr)}{\eps}}$.
Therefore, since we require $\log T = \O{r\log n + \log\log m}$ bits of precisions, our algorithm stores $\frac{n}{\eps'^2}\log n \cdot \poly(\log r \log \log \frac{m}{\delta'}) = \frac{n}{\eps^2} \cdot \poly(\log n, \log r, \log \log \frac{m}{\delta})$ hyperedges.
In addition, by \thmref{thm:robust:graph} the robust graph sparsification algorithm for $G$ with parameters $\tilde{\eps},\tilde{\delta}$ stores $\frac{n r^4}{\eps^2} \cdot \poly(\log n, \log \log \frac{m}{\delta})$ edges.
\end{proof}

\subsection{Graph Min-Cut Approximation}
\seclab{sec:min:cut}
In this section, we apply the streaming framework to solve the graph min-cut approximation problem, which asks for a $(1+\eps)$-approximation to the size of the min-cut. 
First, we introduce a relaxation to the graph spectral sparsifier called the graph for-each spectral sparsifier, which is a graph $\widehat{G}$ that satisfies, for any given vector $x \in \mathbb{R}^n$, it preserves the energy of $x$ in the original graph $G$ with probability $\delta$. 

\begin{theorem}[Offline algorithm for graph for-each sparsifier]
\thmlab{thm:graph:offline}\cite{DingGLLNSW24}
Given a graph $G =(V, E,w)$ with $n$ vertices and $m$ edges, there exists an algorithm that constructs a $(1+\eps)$-for-each spectral sparsifier with probability $1-\frac{1}{\poly(n)}$. It samples $\tO{\frac{n}{\eps}}$ edges and uses $\tO{\frac{n}{\eps}}$ words of working memory and $\tO{m}$ update time.
\end{theorem}

\cite{DingGLLNSW24} proposed an algorithm for min-cut approximation in the offline setting using the above procedure and graph for-all sparsifiers: 
We first utilize the graph spectral sparsifier to give a $2$-approximation to all cuts. 
Then, let $c^*$ denote the minimum cut in the sparsifier, we select a set $S$ that contains all cuts within a factor of $4$ of $c^*$, so $S$ contains the actual minimum cut. 
It is known that there are at most $n^{\O{C}}$ cuts that are within a factor of $C$ of the min-cut \cite{Karger00} for any $C \geq 1$, which means that $S$ has at most $n^{\O{1}}$ items. 
Thus, we can obtain a $(1+\eps)$-graph for-each sparsifier and select the cut in $S$ that has the minimum energy on the for-each sparsifier. 
We can union bound over the failure events of the graph for-each sparsifier across each cut query $S$, so we obtain a $(1+\eps)$-estimation to the min-cut with high probability.
\begin{theorem}[Offline algorithm for min-cut approximation]
\thmlab{thm:min:cut:offline}
\cite{DingGLLNSW24}
Given a graph $G =(V, E,w)$ with $n$ vertices and $m$ edges, there exists an algorithm that provides a $(1+\eps)$-approximation to the min-cut with probability $1-\frac{1}{\poly(n)}$. It samples $\tO{\frac{n}{\eps}}$ edges and uses $\tO{\frac{n}{\eps}}$ words of working memory and $\poly(n)$ update time.
\end{theorem}

\cite{DingGLLNSW24} applied merge-and-reduce directly to the input data stream to achieve a streaming algorithm, which loses an extra $\log n$ factor.
We show that using the aforementioned framework, where we first obtain an online graph spectral sparsifier and then run the offline algorithm in each coreset, we can improve this $\log n$ factor to a $\log \frac{1}{\eps}$ while maintaining a $\poly(n)$ update time, achieving more efficient space.
\begin{theorem}[Streaming algorithm for graph min-cut approximation]
Given a graph $G =(V, E, w)$ with $n$ vertices and $m$ edges defined by an insertion-only stream, there exists a streaming algorithm that provides a $(1+\eps)$-approximation to the min-cut with probability $1-\frac{1}{\poly(n)}$. It stores $\frac{n}{\eps} \cdot \polylog(n,\frac{1}{\eps})$ edges, i.e., $\frac{n}{\eps} \cdot \polylog(n,\frac{1}{\eps})$ bits, and uses $\poly(n)$ update time. 
\end{theorem}
\begin{proof}
We use the online algorithm in \thmref{thm:graph:online} to generate the stream $\calS'$.
It produces an online for-all sparsifier, which generalizes an online cut sparsifier for the graph.
We use the offline algorithm for for-each sparsifier in \thmref{thm:graph:offline} to construct the coresets.
Then, the coreset on the roof of the binary tree is a $(1+\eps)$-for-each sparsifier by setting $\eps' = \eps / h$.
Then, we can use the same method in \thmref{thm:min:cut:offline} to select a set $S$ that contains the min-cut and output the cut with the minimum energy on the $(1+\eps)$-for-each sparsifier.
The union bound argument in \thmref{thm:min:cut:offline} gives the correctness of our streaming algorithm.

The calculation of the space complexity and the update time is similar to that of \thmref{thm:hyper:streaming}. 
From \thmref{thm:graph:offline} and \thmref{thm:graph:online}, we have $|\calS'| = \O{\frac{n}{\eps^2}\log^2n}$ and $|C| = \Theta\left(\frac{n}{\eps}\polylog \frac{n}{\eps}\right)$, so the height of the merge-and-reduce data structure is $h = \log \frac{|\calS'|}{|C|} = \O{\log \frac{1}{\eps}}$.  
Thus, we require $\O{|C| \cdot h} = \frac{n}{\eps} \cdot \polylog(n,\frac{1}{\eps})$ words of space in total. 
The update time of the streaming algorithm is upper bounded by the bounds in \thmref{thm:graph:online} and \thmref{thm:graph:offline}.
\end{proof}

\section{Sliding Window Model}
\seclab{sec:sliding:window}
In this section, we consider the sliding window model.
Although the merge-and-reduce procedure produces a coreset for an insertion-only stream in a straightforward way, it fails for the sliding window model due to the expiration of elements at the beginning of the data stream by the sliding window. 
Since coresets at earlier blocks of the streams are no longer valid, then the coreset at the root of the stream would no longer be accurate. 
To resolve this issue, \cite{WoodruffZZ23} observes that we can once again partition the stream into blocks consisting of $S\left(n,\log m,\frac{\eps}{2\log(mn)},\frac{\delta}{\poly(mn)}\right)$ hyperedges. 
However, instead of creating an offline coreset for the hyperedges in each block of updates, we create an online coreset for the elements in the reverse order of their arrival. 
Specifically, as the hyperedges in each block and each coreset are explicitly stored, we can create a synthetic data stream that consists of the hyperedges in the reverse order, and then we can feed the synthetic stream as input to the online coreset construction. 
In this way, this effectively reverses the stream, so that the sliding window always corresponds to the beginning of the stream. 
Crucially, the online coreset construction implies correctness over any prefix of the reversed stream, which translates to correctness over any suffix of the input stream, including the sliding window. 
For the sake of completeness, we present this approach in \algref{alg:hypergraph:sw}. 

\begin{algorithm}[!htb]
\caption{Algorithm for hypergraph sparsification in the sliding window model via merge-and-reduce and online coresets, adapted from~\cite{WoodruffZZ23}}
\alglab{alg:hypergraph:sw}
\begin{algorithmic}[1]
\State{{\bf Require:} Hyperedges $e_1,\ldots,e_m$ on $n$ vertices, accuracy parameter $\eps\in(0,1)$, failure probability $\delta\in(0,1)$, and window size $W>0$}
\State{{\bf Ensure:} Hypergraph sparsification of the $W$ most recent hyperedges}
\State{Let $\coreset(H,\eps,\delta)$ be an online coreset construction that samples $S(n,\log(mn),\eps,\delta))$ hyperedges, where $H$ has $n$ vertices and $m$ hyperedges}
\State{$M\gets\O{S\left(n,\log(mn),\frac{\eps}{2\log(mn)},\frac{\delta}{(mn)^2}\right)}$}
\State{Initialize coresets $C_0,C_1,\ldots,C_{\log(mn)}\gets\emptyset$}
\For{each hyperedge $e_t$ with $t\in[m]$}
\If{$C_0$ does not contain $M$ hyperedges}
\State{Prepend $e_t$ to $C_0$, i.e., $C_0\gets\{e_t\}\cup C_0$}
\Else
\State{Let $i$ be the smallest index such that $C_i=\emptyset$}
\State{$C_i\gets\coreset\left(\widetilde{H},\frac{\eps}{2\log(mn)},\frac{\delta}{(mn)^2}\right)$, where $\widetilde{H}=C_0\cup\ldots\cup C_{i-1}$}
\State{}
\Comment{$\widetilde{H}$ is an ordered set of weighted hyperedges}
\State{For $j=0$ to $j=i-1$, reset $C_j\gets\emptyset$}
\State{$C_0\gets\{e_t\}$}
\EndIf
\State{{\bf Return} the ordered set $C_0\cup\ldots\cup C_{\log(mn)}$, in reverse}
\EndFor
\end{algorithmic}
\end{algorithm}
The following proof shows the correctness of the framework in \algref{alg:hypergraph:sw}. 
It uses induction and is entirely standard, adapting the approach in \cite{WoodruffZZ23}. 
\begin{theorem}
\thmlab{thm:sw:framework}
Let $e_1,\ldots,e_m$ be a stream of hyperedges, let $\eps\in(0,1)$ be an approximation parameter, and let $H=\{e_{m-W+1},\ldots,e_m\}$ be a hypergraph defined by the $W$ most recent hyperedges. 
Suppose there exists a randomized algorithm that with probability at least $1-\delta$, outputs an online coreset algorithm for hypergraph sparsification using $S(n,\log m,\eps,\delta)$ hyperedges. 
Then there exists a randomized algorithm that with probability at least $1-\delta$, outputs a $(1+\eps)$-hypergraph sparsification in the sliding window model, using $\O{S\left(n,\log m,\frac{\eps}{2\log(mn)},\frac{\delta}{(mn)^2}\right)\log(mn)}$ hyperedges.  
\end{theorem}
\begin{proof}
Consider \algref{alg:hypergraph:sw}, where $\coreset(H,\eps,\delta)$ is a randomized algorithm that, with probability at least $1-\delta$, computes a $(1+\eps)$-approximate online coreset for hypergraph sparsification on an input hypergraph $H$ that has $n$ nodes and $m$ hyperedges. 

We first claim that for each index $i$, $C_i$ is a $\left(1+\frac{\eps}{2\log(mn)}\right)^i$ online coreset for hypergraph sparsification for $2^{i-1}M$ hyperedges. 
Indeed, note that $C_i$ can only be non-empty if at some time, the coreset $C_0$ contains $M$ hyperedges and the coresets $C_1,\ldots,C_{i-1}$ are all non-empty. 
Then by the correctness of the subroutine $\coreset$, $C_i$ is a $\left(1+\frac{\eps}{2\log(mn)}\right)$ online coreset for the hyperedges in $C_0\cup\ldots\cup C_{i-1}$ at some point during the stream. 
It follows that by induction, $C_i$ is a $\left(1+\frac{\eps}{2\log(mn)}\right)\left(1+\frac{\eps}{2\log(mn)}\right)^{i-1}=\left(1+\frac{\eps}{2\log(mn)}\right)^i$ online coreset for $M+\sum_{j=1}^{i-1} 2^{j-1}M=2^{i-1}M$ hyperedges. 

Observe that \algref{alg:hypergraph:sw} inserts the latter hyperedges to the beginning of $C_0$. 
Hence, the stream is fed in reverse to the merge-and-reduce procedure. 
In other words, for any $W\in[2^{i-1},2^i)$, the reverse of $C_0\cup\ldots\cup C_i$ provides a $(1+\eps)$-hypergraph sparsifier for the $W$ most recent hyperedges in the data stream. 

Moreover, there are at most $m$ hyperedges in the data stream. 
For each hyperedge, there are at most $\log(mn)$ coresets constructed by the subroutine $\coreset$, corresponding to the height of the tree. 
Because each subroutine has failure probability at most $\frac{\delta}{(mn)^2}$, then the total failure probability is at most $\delta$ by a union bound.  
This completes the argument for correctness. 

It remains to justify the space complexity. 
To that end, observe that there are at most $\O{\log(mn)}$ online coreset constructions $C_0,\ldots,C_{\log (mn)}$ simultaneously stored by the algorithm. 
Since each online coreset construction samples $S\left(n,\log m,\frac{\eps}{2\log(mn)},\frac{\delta}{(mn)^2}\right)$ hyperedges, then the total number of stored hyperedges is $\O{S\left(n,\log m,\frac{\eps}{2\log(mn)},\frac{\delta}{(mn)^2}\right)\log(mn)}$, as claimed.  
\end{proof}

We give the following result for hypergraph sparsification in the sliding window using the above framework.
\begin{theorem}
Given a hypergraph $H=(V,E)$ with $n$ vertices, $m$ hyperedges, and rank $r$ defined by an insertion-only stream, there exists an algorithm that constructs a $(1+\eps)$-spectral sparsifier in the sliding window model with probability $1-\frac{1}{\poly(n)}$. It stores $\frac{n}{\eps^2}\polylog(m, r)$ hyperedges, i.e., $\frac{rn}{\eps^2}\log n\polylog(m, r)$ bits. 
\end{theorem}

\def\shortbib{0}
\bibliographystyle{alpha}
\bibliography{references}

\appendix
\section{Experiments}
\applab{sec:exp}
In this section, we perform a number of empirical evaluations to complement our theoretical results. 
The experiments are conducted using an Apple M2 CPU, with 16 GB RAM and 8 cores. 

\paragraph{Experiment Setup.}
We compare the performance among three algorithms: the online algorithm, the merge-and-reduce approach applying directly to the input stream, and our streaming algorithm.
In each comparison, we set a budget for the number of sampled edges and compare the multiplicative error of the three algorithms.
For each experiment, we iterate $5$-$10$ times and take the arithmetic mean.
We first test synthetic graphs.
Given inputs $n$ and $m$, we randomly generate a graph with $n$ vertices and $m$ edges: for each edge, we uniformly sample two vertices from $[n]$ and assign its weight from the distribution $\calU(1,10)$.
We allow multi-edges in the graph.
We then test the Facebook ego social network from the Stanford Large Network Dataset Collection (SNAP) \cite{McAuleyL12}. 
We choose the graph from user 107 with $n=1034$ and $m=53498$, and we also select its weight from the distribution $\calU(1,10)$ since the original graph is unweighted.

\paragraph{Graph Metric.}
Let $L_G$ and $\widehat{L_G}$ be the graph Laplacian of the original graph and the sparsified graph, respectively.
We note that the multiplicative error of the sparsifier is $\max_{x \in \mathbb{R}^n} \frac{x^T(L_G - \widehat{L_G})x}{x^T L_G x}$, which is the generalized Rayleigh quotient of matrices $L_G$ and $\widehat{L_G}$.
To measure this quantity, we solve the generalized eigenvalue problem $(L_G - \widehat{L_G})x = \lambda\cdot L_G x$ and obtain the maximum eigenvalue, which equals the multiplicative error by the properties of generalized Rayleigh quotient.

\paragraph{Fine-tuning the parameters.}
Given a budget $l$, we fine-tune the parameters for each method such that they output roughly $l$ edges, and thus ensure that they have the same space budget.
In the online algorithm, we sample each edge $e_t$ with probability $\rho \cdot r_{e_t}$, where $r_{e_t}$ is the effective resistance of $e_t$ in the graph $L_G(t)$ constructed by the previously arrived edges.
Since we do not have a fixed relationship between the sum of online leverage scores and $n$, we fine-tune the parameter $\rho$ such that the mean of $10$ trials is within $l \pm 200$.
Note that the second approach and our streaming algorithm include implementing offline algorithm to construct the merge-and-reduce coresets.
In the offline algorithm, we sample each edge $e$ with probability $\rho \cdot r_e$, where $r_e$ is its effective resistance.
Due to the equivalence between the effective resistance and the leverage score of the incidence matrix, we have $\sum_{e \in G} r_e = n$, and so the expected number of sampled edges is $\rho n$.
Thus, we set $\rho = l/n$ such that we sample $l$ edges in expectation, and the actual number of samples only varies a little due to the concentration.
Then, we use the same parameter $\rho$ for the offline algorithm in the second approach and our streaming algorithm to ensure fair comparison.
In our streaming algorithm, recall that we run an online algorithm to obtain the prefix substream $\calS'$ (which is not stored); we tune the length of $\calS'$ to have the optimal error bound.
We explicitly list our choice of parameters in \appref{sec:param}.

We adapt the code for efficiently estimating the effective resistances in the experiments of \cite{chenYVBDMT23}.
In the online algorithm, we do a batch implementation to save time.
We partition the stream into batches with $100$ edges and use the same $L_G(t)$ to compute the effective resistances for a batch, so we do not have to re-construct the Laplacian within a batch.

\paragraph{Results and Discussion.}
In the first experiment, we fix the number of vertices and edges in the graph and compare the performance under different budgets $l$.
For the synthetic graph, we set $n=100, m=50000$, and $l \in \{500, 1000, 1500, 2000, 2500, 3000\}$.
For the Facebook graph, we set $l \in \{10000, 15000, 20000, 25000\}$.
The results are displayed in \figref{fig:budget}.

\begin{figure}[!htb]
\centering
\begin{subfigure}{0.49\textwidth}
    \centering
    \includegraphics[width=\textwidth]{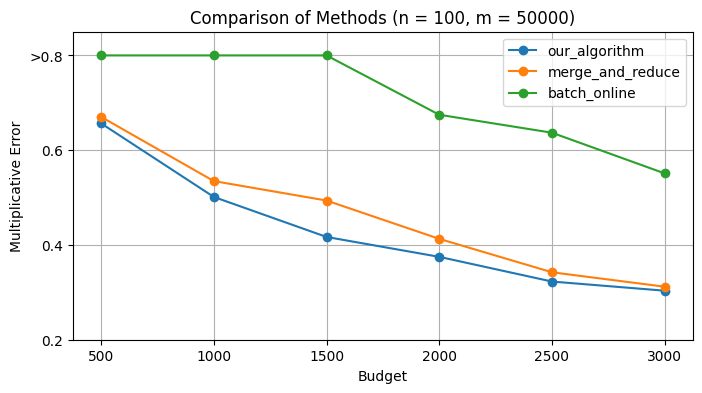}
\end{subfigure}\hfill
\begin{subfigure}{0.49\textwidth}
    \centering
    \includegraphics[width=\textwidth]{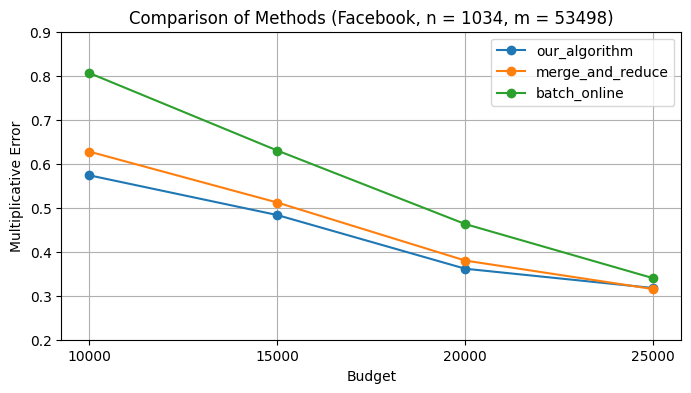}
\end{subfigure}
\caption{Comparison under different budgets. The $x$-axis shows the space budget and the $y$-axis shows the multiplicative error. The left is the synthetic graph and the right is the Facebook graph. We include the result of the online algorithm as a baseline.}
\figlab{fig:budget}
\end{figure}

For both graphs, the merge-and-reduce methods obtain $\sim 0.3$ multiplicative error using $< 50 \%$ budget, and our streaming algorithm (blue line) has the lowest multiplicative error under various budgets.
This occurs because our algorithm applies the merge-and-reduce method to the online substream $\calS'$, resulting in a shorter tree height compared to directly applying merge-and-reduce to the original stream. 
Since errors accumulate at each tree level, our algorithm generally outperforms others.
However, its advantage diminishes as the budget increases. 
This is because with large budgets, we set large coreset sizes; and so the heights of the two methods are roughly the same, reducing the advantage of our algorithm. 
In addition, at extremely low budgets (e.g., $l=500$ in the synthetic graph), our algorithm performs similarly to the merge-and-reduce algorithm.
This implies a trade-off between the accuracy of the online substream and the tree height, when the budget is extremely low compared to $n$.
If we set a small budget for $\calS'$, the accuracy of the prefix online algorithm is suboptimal; on the contrary, if we set a large budget for $\calS'$, then it increases the tree height.
Both situations hinder the performance of our algorithm.

\begin{figure}[!htb]
    \centering
    \includegraphics[width=0.7\linewidth]{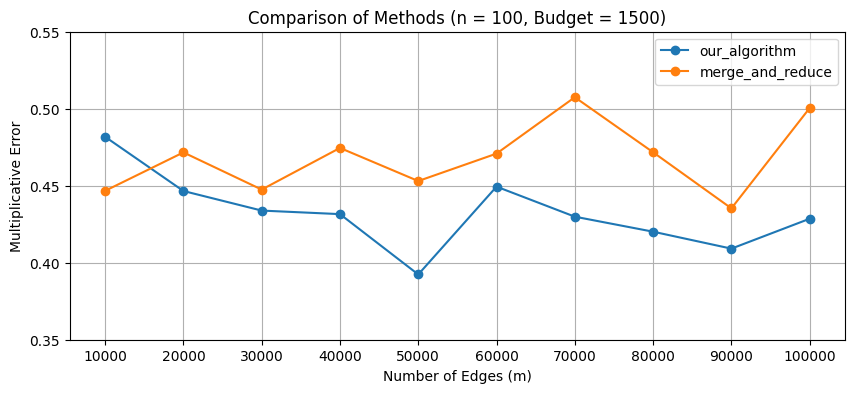}
    \caption{Comparison under different numbers of edges}
    \figlab{fig:edge}
\end{figure}
In the second experiment, we fix the budget and compare the performance under different numbers of edges $m$.
We use the synthetic graph, and we set $n=100$, $l=1500$, and $m \in \{i \cdot 10000, i \in [10]\}$.
The results are displayed in \figref{fig:edge}, showing that our algorithm generally outperforms the merge-and-reduce method as $m$ increases.
This advantage is because we set our budget as $l=1500$, which is substantial compared to $n$, enabling the prefix online substream $\calS'$ to have a high-quality approximation while maintaining a small tree height. 
Consequently, this enhances our algorithm’s performance. 
In summary, our experiments demonstrate the optimality of our streaming algorithm under dense graphs and limited budgets, matching the theoretical guarantees.

\subsection{Additional Experimental Setup}
\applab{sec:param}

In this section, we provide the choices of the hyperparameters of our experiments described in \appref{sec:exp}. 
We use $C_{\text{ol}}$, $C_{\text{of}}$, and $C_{\text{ol\_str}}$ to denote the constant factors of the online algorithm, the offline algorithm that constructs the coresets, and the prefix online substream in our streaming algorithm, respectively.
Their values are shown in \tableref{tab:param}.

\begin{table}[!ht]
\centering

\begin{tabular}{lcccccc}
\toprule
\multicolumn{7}{c}{\textit{Synthetic Graph ($n=100$, $m=50000$)}} \\
\cmidrule(lr){1-7}
Budget & 500 & 1000 & 1500 & 2000 & 2500 & 3000 \\
\midrule
$C_{\text{ol}}$ & 0.001 & 0.01 & 0.15 & 0.35 & 0.55 & 0.75 \\
$C_{\text{off}}$ & 1.1 & 2.2 & 3.3 & 4.4 & 5.5 & 6.6 \\
$C_{\text{ol\_str}}$ & 2.0 & 2.5 & 5.5 & 8.0 & 11.5 & 15.0 \\
\bottomrule
\end{tabular}

\vspace{0.5cm}

\begin{tabular}{lcccc}
\toprule
\multicolumn{5}{c}{\textit{Facebook Graph ($n=1034$, $m=53498$)}} \\
\cmidrule(lr){1-5}
Budget & 5000 & 7500 & 10000 & 12500 \\
\midrule
$C_{\text{ol}}$ & 0.05 & 0.105 & 0.145 & 0.185 \\
$C_{\text{off}}$ & 0.69 & 1.0 & 1.333 & 1.667 \\
$C_{\text{ol\_str}}$ & 0.8 & 0.8 & 1.0 & 1.4 \\
\bottomrule
\end{tabular}

\vspace{0.5cm}

\small
\setlength{\tabcolsep}{3pt}
\begin{tabular}{l*{10}{c}}
\toprule
\multicolumn{11}{c}{\textit{Comparison under Different $m$, Synthetic Graph ($n=100$, Budget=$1500$)}} \\
\cmidrule(lr){1-11}
m & 10000 & 20000 & 30000 & 40000 & 50000 & 60000 & 70000 & 80000 & 90000 & 100000 \\
\midrule
$C_{\text{ol}}$ & 0.3 & 0.2 & 0.15 & 0.1 & 0.08 & 0.05 & 0.04 & 0.02 & 0.005 & 0.002 \\
$C_{\text{off}}$ & 3.3 & 3.3 & 3.3 & 3.3 & 3.3 & 3.3 & 3.3 & 3.3 & 3.3 & 3.3 \\
$C_{\text{ol\_str}}$ & 15.0 & 7.0 & 5.5 & 4.8 & 4.3 & 4.05 & 3.75 & 3.75 & 3.7 & 3.2 \\
\bottomrule
\end{tabular}
\vspace{0.5cm}

\caption{Constant factors ($C_{\text{ol}}$, $C_{\text{off}}$, $C_{\text{ol\_str}}$) for the three experiments.}
\tablelab{tab:param}
\end{table}

We also include the (amortized) run-time of our algorithm in \tableref{tab:running-times}, showing that our algorithm is time-efficient in practice.

\begin{table}[!ht]
\centering
\begin{tabular}{lcccccc}
\hline
\textbf{Note} & $n$ & $m$ & \textbf{Budget} & \textbf{Batched online} & \textbf{Merge-and-reduce} & \textbf{Our algorithm} \\
\hline
Synthetic graph & 50   & 2500  & 750   & $2.0 \times 10^{-4}\,s$ & $2.9 \times 10^{-4}\,s$ & $4.9 \times 10^{-4}\,s$ \\
Synthetic graph & 100  & 10000 & 2500  & $6.1 \times 10^{-4}\,s$ & $3.0 \times 10^{-4}\,s$ & $5.6 \times 10^{-4}\,s$ \\
Synthetic graph & 200  & 40000 & 10000 & $7.6 \times 10^{-4}\,s$ & $4.7 \times 10^{-4}\,s$ & $6.1 \times 10^{-4}\,s$ \\
Facebook graph  & 1034 & 53498 & 15000 & $4.0 \times 10^{-2}\,s$ & $2.3 \times 10^{-2}\,s$ & $9.6 \times 10^{-2}\,s$ \\
\hline
\end{tabular}
\caption{Comparison of running times in different experiments.}
\tablelab{tab:running-times}
\end{table}

\end{document}